\documentclass[10pt,twocolumn,letterpaper]{article}

\usepackage{cvpr}              % To produce the CAMERA-READY version
\usepackage[accsupp]{axessibility}

\usepackage{balance}
\usepackage{url}            % simple URL typesetting
\usepackage{booktabs}       % professional-quality tables
\usepackage{amsfonts}       % blackboard math symbols
\usepackage{nicefrac}       % compact symbols for 1/2, etc.
\usepackage{microtype}      % microtypography
\usepackage{xcolor}         % colors
\usepackage{graphicx}
\usepackage{bm}
\usepackage{amsmath}
\usepackage{mathtools}
\usepackage{amsthm}
\usepackage{algorithm}
\usepackage[noend]{algorithmic}
\usepackage{wrapfig}
\usepackage{subcaption}
\usepackage{multirow}
\usepackage{siunitx}

\usepackage{enumitem}
\usepackage{tikz}
\usepackage{pgfplots}
\tikzset{every picture/.style={/utils/exec={\sffamily}}}

\newcommand{\red}[1]{\textcolor{red}{#1}}

\newcommand{\blue}[1]{\textcolor{blue}{#1}}

\newcommand{\method}{\texttt{Delta}}
\newcommand{\IRM}{$\text{IR}_{\text{main}}$}
\newcommand{\IRR}{$\text{IR}_{\text{res}}$}
\newcommand{\IRN}{$\text{IR}_{\text{noisy}}$}
\newcommand{\IRQ}{$\text{IR}_{\text{quant}}$}
\newcommand{\MM}{$\mathcal{M}_{\text{main}}$}
\newcommand{\MRS}{$\mathcal{M}_{\text{res}}$}
\newcommand{\MB}{$\mathcal{M}_{\text{bb}}$}
\newcommand{\zM}{$\bm{z}_{\text{main}}$}
\newcommand{\zRS}{$\bm{z}_{\text{res}}$}

\newcommand{\norm}[1]{$\left \| #1 \right \|_F$}
\newcommand{\difffro}[2]{$\frac{\left \| #1-#2 \right \|}{\left \| #1 \right \|}$}

\theoremstyle{plain}
\newtheorem{theorem}{Theorem}
\newtheorem{appxtheorem}{Theorem}

\theoremstyle{definition}

\theoremstyle{remark}
\newtheorem{remark}{Remark}

\newcommand{\attMMI}{\texttt{ML-Leaks}}
\def\eqref#1{equation~\ref{#1}}
\def\1{\bm{1}}

\def\vo{{\bm{o}}}

\def\vu{{\bm{u}}}
\def\vv{{\bm{v}}}

\def\vx{{\bm{x}}}

\DeclareMathAlphabet{\mathsfit}{\encodingdefault}{\sfdefault}{m}{sl}
\SetMathAlphabet{\mathsfit}{bold}{\encodingdefault}{\sfdefault}{bx}{n}
\newcommand{\tens}[1]{\bm{\mathsfit{#1}}}

\def\tI{{\tens{I}}}

\def\tN{{\tens{N}}}

\def\tW{{\tens{W}}}
\def\tX{{\tens{X}}}
\def\tY{{\tens{Y}}}

\def\sR{{\mathbb{R}}}

\newcommand{\etens}[1]{\mathsfit{#1}}

\def\etC{{\etens{C}}}

\def\etU{{\etens{U}}}

\def\etX{{\etens{X}}}

\usepackage{environ}
\pgfplotsset{compat = newest}
\usetikzlibrary{arrows,positioning} 
\usetikzlibrary{decorations.text}
\usetikzlibrary{decorations.pathreplacing, calc}
\usetikzlibrary{arrows.meta}
\tikzset{
    >=stealth',
    port/.style = {circle, draw, align=center, minimum height=1mm},
    op/.style={
           rectangle,
           rounded corners,
           draw=black, thick,
           text width=3.5em,
           minimum height=1em,
           text centered},
    conv/.style={
           rectangle,
           rounded corners,
           draw=black, thick,
           text width=3.5em,
           minimum height=1em,
           text centered},
    data/.style={
           rectangle,
           draw=black, thick,
           minimum height=1em,
           text centered},
    area/.style={
           rectangle,
           draw=black, thick,
           minimum height=1em},
    connect/.style={
           ->,
           thick,
           shorten <=2pt,
           shorten >=2pt,}
}

\definecolor{cvprblue}{rgb}{0.21,0.49,0.74}
\usepackage[pagebackref,breaklinks,colorlinks,citecolor=cvprblue]{hyperref}

\def\paperID{13752} % *** Enter the Paper ID here
\def\confName{CVPR}
\def\confYear{2024}

\title{All Rivers Run to the Sea: Private Learning with Asymmetric Flows}

\author{Yue Niu$^{1}$\quad
Ramy E. Ali$^{2}$\thanks{R.E. Ali was with the University of Southern California.} \quad
Saurav Prakash$^{3}$\quad
Salman Avestimehr$^{1}$\\
$^{1}$ University of Southern California \quad
$^{2}$ Samsung \quad
$^{3}$ University of Illinois Urbana-Champaign \\
{\tt\small yueniu@usc.edu \quad ramy.ali@samsung.com \quad sauravp2@illinois.edu \quad avestime.usc.edu}
}
\begin{document}
\maketitle
\begin{abstract}
Data privacy is of great concern in cloud machine-learning service platforms, when sensitive data are exposed to service providers. While private computing environments (e.g., secure enclaves), and cryptographic approaches (e.g., homomorphic encryption) provide strong privacy protection, their computing performance still falls short compared to cloud GPUs. To achieve privacy protection with high computing performance, we propose \method{}, a new private training and inference framework, with comparable model performance as non-private centralized training. \method{} features two asymmetric data flows: the main information-sensitive flow and the residual flow. The main part flows into a small model while the residuals are offloaded to a large model. Specifically, \method{} embeds the information-sensitive representations into a low-dimensional space while pushing the information-insensitive part into high-dimension residuals. To ensure privacy protection, the low-dimensional information-sensitive part is secured and fed to a small model in a private environment. On the other hand, the residual part is sent to fast cloud GPUs, and processed by a large model. To further enhance privacy and reduce the communication cost, \method{} applies a random binary quantization technique along with a DP-based technique to the residuals before sharing them with the public platform. We theoretically show that \method{} guarantees differential privacy in the public environment and greatly reduces the complexity in the private environment. We conduct empirical analyses on CIFAR-10, CIFAR-100 and ImageNet datasets and ResNet-18 and ResNet-34, showing that \method{} achieves strong privacy protection, fast training, and inference without significantly compromising the model utility.
\end{abstract}    

\section{Introduction}
\label{sec:intro}
In the current machine learning (ML) era, cloud ML services with high-end GPUs have become indispensable. 
On the other hand, ensuring data privacy is one of the most critical challenges in the ML platforms. 
During training, privacy breaches may occur if training data is exposed to ML service providers, increasing vulnerability to potential attacks. 
% Moreover, knowing certain training information such as parameters, gradients, or a model's outputs, attack methods such as membership inference attacks \cite{MembershipInfer_SP2017, MLLeaks, labelMI} and model inversion attacks \cite{ModelInversion_CCS2015, GradInversion_NIPS2019} can infer the membership or even reconstruct training samples. 
Additionally, users' inference queries can also be susceptible to attacks when accessing ML services with sensitive data \cite{ChatGPT, Midjourney}. 
In particular, an untrusted ML platform can cache, learn, and leak queries without users' awareness \cite{ChatGPTLeaking}. 
% Sometimes, the service provider may even expose users' data to a third party, who may misuse users' data for other purposes.

\textbf{Related Works Overview}. While prior privacy-preserving machine learning (PPML) frameworks \cite{DLDP_SP2016, DLFHE_NIPS2016, DarkNet_MobiSys2020, FL_AISTAT2017, PriSM_TMLR_2023, Flash_TMLR_2023} mitigate privacy concerns in training and inference, they also come with different tradeoffs.
Differential privacy (DP) based methods perturb the data before outsourcing to an untrusted cloud \cite{Shredder_ASPLOS2020, ML_Obfuscation} to ensure privacy, but
% For protection after model release, methods such as the DP-SGD \cite{DLDP_SP2016} inject noise to gradients to avoid model overfitting to the training data. 
they usually result in degraded model utility even under moderate privacy constraints \cite{DPACC_NIPS2019}.
% However, such methods all suffer very poor privacy-utility tradeoffs, and are barely practical in real systems. \\
% Machine learning with differential privacy (DP) \cite{DLDP_SP2016, ScaleDP_PETS2021, Shredder_ASPLOS2020}, for instance, though to some extent protects data from model or gradient inversion attacks, usually incurs a significant accuracy drop \cite{DPACC_NIPS2019}. 
On the other hand, crypto-based techniques \cite{DLFHE_NIPS2016, DarkNet_MobiSys2020}, that provide data protection with encryption schemes, have not yet proved efficient and scalable to large models due to their prohibitive complexities. 

PPML with private environments (e.g., trusted execution environments (TEEs), local environments) presents a promising solution by physically isolating the running computing environments. 
Such private environments are, however, usually resource-constrained compared to public cloud services with high-end GPUs, resulting in lower computing performance \cite{SecureTF_Mid2020, Capsule_CVPR2021}. 
However, the prior works based on leveraging these private environments along with the  public GPUs also incur high complexity \cite{Slalom_ICLR_20, legrace_PETS_2022, Origami_CLOUD_2021, DarkNight_MICRO_2021}. One reason for these methods being inefficient is the heavy communication between the private and the public environments. 

\begin{figure*}[!htb]
    \centering
    \begin{subfigure}{.2\textwidth}
        \centering
        \begin{tikzpicture}
            \node[](USER){\includegraphics[width=7mm]{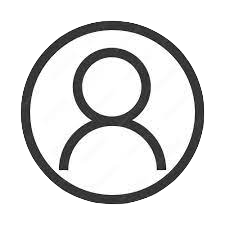}};
            \node[above=.1mm of USER](DATA){\includegraphics[width=5mm]{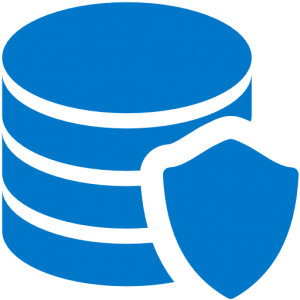}};
            \node[above=.1mm of DATA](MODEL){\includegraphics[width=5mm]{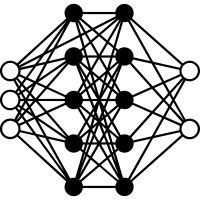}};
            % \node[below=1mm of USER, xshift=10mm](){\scriptsize User's private data and model};
            \node[op, right=15mm of USER, draw=pink!50, fill=pink!30, text width=10mm](PRIV){\scriptsize Private Env};
            \node[op, above=10mm of PRIV, draw=black!50, fill=gray!20, text width=10mm](PUB){\scriptsize Public Env};
            \node[right=0.1mm of PRIV](){\includegraphics[width=6mm]{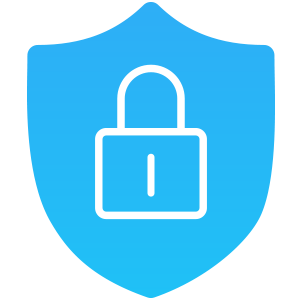}};
            \node[right=0.1mm of PUB](){\includegraphics[width=6mm]{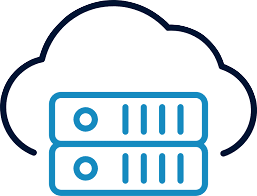}};

            \draw[->, red!50, thick] (USER.east) to node[above] {\scriptsize Private} (PRIV.west);
            \draw[->, black!50] (USER.north east) to [out=30, in=240] node[above, xshift=-1mm] {\scriptsize Public} (PUB.south west);
            \draw[->, black!50, dashed] (PRIV.north) to (PUB.south);
            \draw[->, black!50, dashed] (PUB.south) to (PRIV.north);
        \end{tikzpicture}
        \caption{\footnotesize Problem Setup. }
        \label{fig:overview:problem}
    \end{subfigure}
    \hfill
    \begin{subfigure}{.7\textwidth}
    \centering
    \begin{tikzpicture}
        % private flow
        \node[](dog){\includegraphics[width=.07\textwidth]{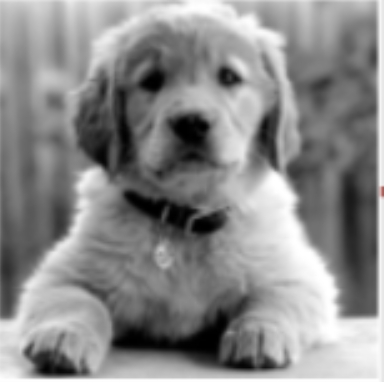}};
        \node[op, draw=pink!50, fill=pink!30, text width=8mm, right=3mm of dog](MB){\scriptsize \MB{}};
        \node[op, draw=pink!50, fill=pink!30, text width=12mm, right=3mm of MB](IR){\scriptsize Decompose};
        \node[op, draw=pink!50, fill=pink!30, text width=8mm, right=12mm of IR](MM){\scriptsize \MM{}};
        \node[op, draw=pink!50, fill=pink!30, text width=0.1mm, minimum height=8mm, right=5mm of MM, label={[text depth=0ex,rotate=-90]center: \tiny logits}](logits1){};
        \node[op, draw=pink!50, fill=pink!30, text width=6mm, right=8mm of logits1](ADD){\scriptsize Add};
        \node[op, draw=pink!50, fill=pink!30, text width=10mm, right=2mm of ADD](SOFTMAX){\scriptsize \texttt{SoftMax}};
        \node[data, draw=pink!50, fill=pink!30, text width=6mm, right=2mm of SOFTMAX](PRED){\scriptsize Pred};
        \node[above=0.1mm of MM, xshift=-3mm](dog2){\includegraphics[width=4mm]{fig/dog.png}};

        \draw[->, red!50, thick] (dog.east) to (MB.west);
        \draw[->, red!50, thick] (MB.east) to (IR.west);
        \draw[->, red!50, thick] (IR.east) to node[below] {\tiny \IRM{}} (MM.west);
        \draw[->, red!50, thick] (MM.east) to (logits1.west);
        \draw[->, red!50, thick] (logits1.east) to (ADD.west);
        \draw[->, red!50, thick] (ADD.east) to (SOFTMAX.west);
        \draw[->, red!50, thick] (SOFTMAX.east) to (PRED.west);

        % public flow
        \node[op, draw=pink!50, fill=pink!30, text width=3mm, right=1mm of IR, yshift=8mm](DP){+};
        \node[op, draw=pink!50, fill=pink!30, text width=6mm, above=1mm of DP](QUANT){\scriptsize Quant};
        % \node[rounded corners, draw=red!30, right=0.1mm of IR, text width=8mm, minimum height=10.5mm, yshift=11mm](NOISE){};
        \node[op, draw=black!40, fill=gray!20, text width=10mm, right=10mm of IR, yshift=24mm, minimum height=8mm](MR){\scriptsize \MRS{}};
        \node[op, draw=black!40, fill=gray!20, text width=0.1mm, minimum height=8mm, right=3mm of MR, label={[text depth=0ex,rotate=-90]center: \tiny logits}](logits2){};
        \node[left=8mm of MR](nois){\includegraphics[width=5mm]{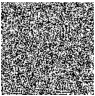}};
        \node[left=2mm of DP](NOISE){\scriptsize noise};

        \draw[->, red!50] (IR.east) to [out=0, in=270] node[above, xshift=-2mm] {\tiny \IRR{}} (DP.south);
        \draw[-, red!50] (DP) to (QUANT);
        \draw[-, red!50](QUANT.north) -- +(0mm, 1.5mm);
        \draw[->, black!50] (QUANT.north) +(0mm, 1.5mm) to [out=90, in=180] node[above, xshift=-1mm] {\tiny \IRQ{}} (MR.west);
        \draw[->, black!50] (MR.east) to (logits2.west);
        \draw[->, black!50] (logits2.east) to [out=0, in=180] (ADD.west);
        \draw[->, red!50, dashed] (NOISE) to (DP);

        % outer box
        \node[rounded corners, draw=red!30, thick, double, text width=120mm, minimum height=24mm, xshift=52mm, yshift=6mm]{};
        \node[above=3mm of dog](LOCAL){\includegraphics[width=6mm]{fig/protect.png}};
        \node[above=15mm of dog](CLOUD){\includegraphics[width=6mm]{fig/cloud.png}};
        \node[right=-1mm of LOCAL, align=left]{\tiny \textbf{Private (TEEs, local env)}};
        \node[right=-1mm of CLOUD, align=left]{\tiny \textbf{Public (cloud GPUs)}};
        
    \end{tikzpicture}
    \caption{\footnotesize An overview of \method{}.}
    \label{fig:overview:delta}
    \end{subfigure}
    \caption{\footnotesize Overview of \method{}: the backbone \MB{} acts as a feature extractor. The features are decomposed into low-dimensional (information-sensitive) and high-dimensional (residual) parts: \IRM{} and \IRR{}. \IRM{} is fed to a small model \MM{}, while \IRR{} are outsourced to a large model \MRS{}. \MB{} and \MM{} run in a resource-constrained private environment, whereas \MRS{} is offloaded to a public environment while ensuring privacy through a DP scheme. While only the forward pass is shown, backpropagation is also private (See Sec \ref{sec:method:backprop}).}
    \label{fig:overview}
    \vspace{-4mm}
\end{figure*}

\textbf{Proposed Solution Overview.} In this paper, we \emph{ consider a private training and inference setting where users can access both private and public environments}, as shown in Figure \ref{fig:overview:problem}. We propose a new private training and inference framework, \method{}, that achieves strong privacy protection with high model utility and low complexity.
The core idea of \method{} originates from an observation that the intermediate representations (IRs) in ML models exhibit an \emph{asymmetric} structure.
Specifically, the primary sensitive information is usually encoded in a low-dimensional space, while the high-dimensional residuals contain very little information. \\
Inspired by such an observation, we design a two-way training framework that respectively learns low-dimensional but information-sensitive features \IRM{} with a \emph{small} model \MM{}, and learns the high-dimensional residuals \IRR{} with a \emph{large} model \MRS{}, as illustrated in Figure \ref{fig:overview:delta}.
Given a model, \method{} selects a few front layers as a backbone (\MB{}) for extracting features and generating intermediate representations.
\method{} uses singular value decomposition (SVD) and discrete cosine transformation (DCT) to extract the low-dimensional information-sensitive representation \IRM{} and the residuals \IRR{}.
We design a new low-dimensional model (\MM{}) to learn the information-sensitive \IRM{}.
While \MB{} and \MM{} run in a private environment, the rest of the model (\MRS{}) learns residuals in a public environment.
\method{} further applies DP perturbation and binary quantization on the residuals, leading to further privacy protection and a communication reduction. \\
Owing to the asymmetric structure in IR, \method{} guarantees differential privacy on \IRR{} with much smaller noise added compared to the naive scheme that directly adds noise to IR (naive-DP), leading to much-improved model utility.
In both training and inference phases, \method{} ensures that only the residual part is public, while the information-sensitive part is secured in private environments. 

\method{} is a generic PPML solution that can be flexibly deployed in several scenarios. 
For instance, in a cloud ML platform with private TEEs and public GPUs \cite{MSEnclave, AWSEnclave}, \MB{} and \MM{} can run inside TEEs to preserve privacy, while \MRS{} runs in GPUs to speed up the computations.
In general distributed settings, \method{} can let clients train \MB{} and \MM{} locally, while a cloud server performs side training with \MRS{}. In summary, our contributions are as follows. 

1. We propose a PPML framework for training and inference with a much-improved privacy-utility trade-off compared to the naive-DP methods, a low computing complexity in the private environments and a low communication cost. 
%\footnote{These include TEEs, client-side CPU, etc.}. 

2. We design an asymmetric decomposition layer that extracts the low-dimensional information-sensitive representations using SVD and DCT, and design a low-complexity model for resource-constrained computing environments.
    
3. We provide a formal differential privacy analysis for the proposed framework. In addition, we show empirically that \method{} provides strong privacy protection against model inversion and membership inference attacks. 

4. We conduct comprehensive evaluations, showing that \method{} leads to a better privacy-utility trade-off than the naive DP-based method that directly adds noise to IR. Specifically, under the same privacy budget, \method{} improves the accuracy by up to $31\%$. Moreover, \method{} greatly reduces the running time compared to other PPML solutions.

\section{System Model}\label{sec:model} 

We start by describing the problem setting, the threat model, and our notations. 

\textbf{Problem Setting.}
We consider a setting where users have private resource-constrained environments (e.g., cloud TEEs, local CPUs/GPUs), but they can also access public, untrusted, and high-end services (e.g., cloud GPUs) to accelerate training and inference. The goal is to protect users' training and inference data, while maintaining computing performance and high model utility.
% We consider an outsourcing learning or inference setting with a private resource-constrained computing environment (e.g., TEE, client-side CPU/GPU, ...) and a powerful untrusted/public computing platform (e.g., cloud GPU). The goal is to perform learning and inference tasks while satisfying the resource constraints and ensuring data privacy.  
Note that this setting differs from DP-SGD \cite{DLDP_SP2016, ScaleDP_PETS2021, GEP_ICLR_2021}, where the goal is to protect the training data from attacks on gradients or models. 

\textbf{Threat Model.}
Our threat model is similar to the model considered in prior works leveraging private environments \cite{legrace_PETS_2022,Slalom_ICLR_20,SecureTF_Mid2020}. Specifically, we assume that the private environment is protected against all untrusted, unauthorized access to data and models inside. However, denial-of-service and side-channel attacks are out of our scope. On the other hand, the untrusted, public environment is semi-honest (honest-but-curious). That is, the untrusted public server follows the training and inference protocol faithfully but may attempt to learn as much as possible from what it receives. 
\begin{remark}
While we only consider the semi-honest model in our work, verifiable computing techniques \cite{thaler2013time,SafetyNets_NIPS_2017,ali2020polynomial} can be incorporated to enhance Delta against malicious parties.
\end{remark}
% (e.g., residuals to the public environment).
%and label information). 

\iffalse
\textbf{Threat Model - } We consider a threat model such that 1) an adversary may access any components (e.g., GPU runtime, servers in FL setup) in an ML system except the private environment; 2) an adversary may obtain model parameters/gradients in the public platform; 3) an adversary may obtain residuals and data in the public environment; 4) an adversary may have some prior knowledge of the training datasets such as labels, and use public resources to improve its attack performance. 
\fi

\textbf{Notations.} 
We denote tensors by capital bold letters as $\tX$, matrices by capital letters as $X$, and vectors by small bold letters as $\vx$.
$\etX^i$ denotes a tensor slice $\tX(i,:)$.
\norm{.} denotes the Frobenius norm or in general the square root of the sum of squares of elements in a tensor.
``$\cdot$" denotes matrix multiplication or in general batch matrix multiplication.
``$\circledast$" denotes a convolution operation.
$X^*$ indicates a transpose.

\section{Asymmetric Structure in IRs}\label{sec:motivation}

We first observe that intermediate representations (IRs) in neural networks (NNs) exhibit highly asymmetric structures in multiple dimensions. 
These asymmetric structures are essential for an asymmetric decomposition that embeds the privacy-sensitive information into low-dimensional space.  
% As we show later in the paper, our \emph{asymmetric learning} framework can exploit such asymmetric structures, which uses a private small model to learn the privacy-sensitive part and outsource residuals to an untrusted fast platform.
% Before that, we first show the asymmetric structure in the channels and spatial dimensions in this section.
% \red{In Transformer architectures?}

\iffalse
\begin{figure}[!htb]
    \centering   
    \begin{subfigure}{.4\textwidth}
        \centering
        \includegraphics[width=\linewidth]{plot/motivation1.pdf}
        \caption{\footnotesize Energy ratio of the low-rank approximation obtained through SVD vs. $r/c \times 100 \%$.}
        \label{fig:motivation:a}
    \end{subfigure}
    \begin{subfigure}{.4\textwidth}
        \centering
        \includegraphics[width=\linewidth]{plot/motivation2.pdf}
        \caption{\footnotesize Energy ratio of the low-frequency approximation obtained through DCT vs. $(t'/t)^2 \times 100 \%$.}
        \label{fig:motivation:b}
    \end{subfigure}
    \caption{Asymmetric structures along the channel and spatial dimension after first convolution layer in ResNet-18 (with ImageNet samples). Most information in $\tX$ can be embedded into low-dimensional representations along channels/spatial dimensions. Experiment details are provided in Appendix \ref{appx:motivation}.}
    \label{fig:motivation}
\end{figure}
\fi

\subsection{Asymmetric Structure in Channel Dimension}
\label{sec:motivation-SVD}
In NNs such as convolutional neural networks (CNNs), each layer's input and output consists of multiple channels, denoted as $\tX\in\sR^{c\times h\times w}$, where $c$ is the number of channels and $h$ and $w$ denote the height and the width.
% These channels usually have a strong correlation.

We analyze the channel correlation by first flattening $\tX$ as $X\in\sR^{c\times hw}$ and then performing singular value decomposition (SVD) on the flattened tensor $X$ as
\begin{center}
$X = \sum\limits_{i=1}^c s_i \cdot \vu_i \cdot \vv_i^*$, 
\end{center}
where $s_i$ is the $i$-th singular value, $\vu_i \in \mathbb R^{c}$ and $\vv_i \in \mathbb R^{hw}$ are the $i$-th left and right singular vectors, respectively. 
We reshape $\vv_i$ to the original dimensions as a \emph{principal channel} $V^i\in\sR^{h\times w}$, and $\vu_i$ as a tensor $\etU^i\in\sR^{c\times 1\times 1}$. \\
We then obtain a \underline{l}ow-\underline{r}ank representation of $\tX$ as follows
\begin{center}
    $\tX_{\text{lr}} = \sum\limits_{i=1}^r s_i \cdot \etU^i \cdot V^i$,
\end{center} 
where $r$ denotes the number of principal channels in $\tX_{\rm{lr}} \in \mathbb R ^{c \times h \times w}$.
Figure \ref{fig:motivation:a} shows the normalized error \difffro{\tX}{\tX_{\text{lr}}} versus $r$ after the first convolutional layer in ResNet-18 (based on ImageNet). 
We observe that $\tX_{\text{lr}}$ with a small $r$ is sufficient to approximate $\tX$.
That is, most information in $\tX$ can be embedded into $\tX_{\text{lr}}$ in a low-dimensional space.
We notice that \texttt{3LegRace} \cite{legrace_PETS_2022} also investigated such a property. However, unlike \texttt{3LegRace}, we leverage the low-rank property to design a new low-complexity model (See Sec. \ref{sec:method:model}). 

\subsection{Asymmetric Structure in Spatial Dimension}
\label{sec:motivation-DCT}
The asymmetric structure of the IR also exists over the spatial dimension due to correlations among pixels in each channel $\etX^i\in\sR^{h\times w}$.
We use discrete cosine transform (DCT) to analyze the spatial correlation. Specifically, we obtain frequency components using $t\times t$ block-wise DCT \cite{BookSigProc} as,
\begin{equation*}
\etC^i = \text{DCT}(\etX^i, T) = \left \{ \text{ } \etC^i_{k,j} = T \cdot \etX^i_{k,j} \cdot T^* \text{ } \right \}_{k,j=1}^{h/t, w/t},
\end{equation*}
where $\etC^i_{k,j}\in\sR^{t \times t}, \etX^i_{k,j} = \etX^i ( \text{ } kt-t:kt,jt-t:jt \text{ } )$ and $T\in\sR^{t\times t}$ is the DCT transformation matrix. 
$\etC^i\in\sR^{h\times w}$ is obtained by simply concatenating $\etC^i_{k,j}$ for $k \in {1, ..., h/t}, j \in {1, ..., w/t}$. Then, we obtain a \underline{l}ow-\underline{f}requency representation, $\etX^i_{\text{lf}}$, using inverse $t' \times t'$ block-wise DCT as,
\begin{equation*}
\etX^i_{\text{lf}} = \text{IDCT} (\etC^i_{\text{lf}}, T_{\text{lf}}) = \left \{ \text{ } \etX^i_{\text{lf},k,j} = T^*_{\text{lf}} \cdot \etC^i_{\text{lf},k,j} \cdot T_{\text{lf}} \text{ } \right \}_{k,j=1}^{h/t, w/t},
\end{equation*}
where $t'<t$, $\etX^i_{\text{lf}}\in\sR^{\frac{h}{t}t' \times \frac{w}{t}t'}$, $\etC^i_{\text{lf},k,j} = \etC^i_{k, j}(0:t', 0:t')$, and $T_{\text{lf}}\in\sR^{t' \times t'}$ is the DCT matrix. Note that $t'^{2}$ represents the number of low-frequency components in $\etX_{\text{lf},k,j}^i$.

\begin{figure}[!htb]
    \centering
    \begin{subfigure}{.45\textwidth}
        \centering
        \begin{tikzpicture}
        \begin{axis}[
            xmin = 0.5, xmax = 8.5, ymin = -0.05, ymax = 0.15,
            xlabel = {\scriptsize Fraction of principal channels in $\tX_{\text{lr}}$},
            ylabel = {\scriptsize Error}, 
            xtick = { 2, 4, 6, 8 },
            xticklabels = {$12.5\%$, $25\%$, $37.5\%$, $50\%$},
            xticklabel style={ align=center,text width=14mm },
            ytick = {0, 0.1 }, yticklabels = {0, 0.1 },
            ymajorgrids, major grid style={line width=.1pt, draw=black!30, dashed},
            width=.9\textwidth, height=.35\textwidth, line width=0.1pt,
            ticklabel style={font=\tiny},
            ]
    
            \addplot+[ 
            thick, blue, smooth,
            mark=*, mark options={solid, scale=0.75},
            error bars/.cd,
              y fixed,
              y dir=both,
              y explicit
            ] 
            table [x=x, y=y, y error plus=error1, y error minus=error2, col sep=comma]{
                x,      y,      error1,     error2
                1,      0.119,   0,       0
                2,      0.047,   0,       0
                3,      0.016,   0,       0
                4,      0.003,   0,       0
                5,      0.001,   0,       0
                6,      0.001,   0,       0
                7,      0.001,   0,       0
                8,      0.001,   0,       0
            };
            % \addlegendentry{\scriptsize $\frac{\left \| \tX - \tX_{\text{lr}} \right \|}{\left \| \tX \right \|}$} 
            \node[fill=white, draw=black!30, rounded corners] at (axis cs:7.5,0.08){\tiny {$\frac{\left \| \tX - \tX_{\text{lr}} \right \|}{\left \| \tX \right \|}$}};
        \end{axis}
        \end{tikzpicture}
        \caption{\footnotesize Error of the low-rank approximation vs. $r/c \times 100 \%$.}
        \label{fig:motivation:a}
    \end{subfigure}

    \vspace{1mm}
    
    \begin{subfigure}{.45\textwidth}
        \centering
        \begin{tikzpicture}
        \begin{axis}[
            xmin = 0.5, xmax = 8.5,  ymin = -0.05, ymax = 0.15,
            xlabel = {\scriptsize Fraction of low-freq components in $\tX_{\text{lf}}$},
            ylabel = {\scriptsize Error}, 
            xtick = { 2, 4, 6, 8 },
            xticklabels = {$8\%$, $18\%$, $32\%$, $50\%$},
            xticklabel style={ align=center,text width=14mm },
            ytick = {0, 0.1 }, yticklabels = {0, 0.1 },
            ymajorgrids, major grid style={line width=.1pt, draw=black!30, dashed},
            width=.9\textwidth, height=.35\textwidth, line width=0.1pt,
            ticklabel style={font=\tiny},
            ]
    
            \addplot+[ 
            thick, blue, smooth,
            mark=*, mark options={solid, scale=0.75},
            error bars/.cd,
              y fixed,
              y dir=both,
              y explicit
            ] 
            table [x=x, y=y, y error plus=error1, y error minus=error2, col sep=comma]{
                x,      y,      error1,     error2
                1,      0.105,   0,       0
                2,      0.072,   0,       0
                3,      0.058,   0,       0
                4,      0.042,   0,       0
                5,      0.022,   0,       0
                6,      0.011,   0,       0
                7,      0.005,   0,       0
                8,      0.001,   0,       0
            };
            % \addlegendentry{\scriptsize $\frac{\left \| \tX - \tX_{\text{lf}} \right \|}{\left \| \tX \right \|}$} 
            \node[fill=white, draw=black!30, rounded corners] at (axis cs:7.5,0.08){\tiny {$\frac{\left \| \tX - \tX_{\text{lf}} \right \|}{\left \| \tX \right \|}$}};
        \end{axis}
        \end{tikzpicture}
        \caption{\footnotesize Error of the low-frequency approximation vs. $(t'/t)^2 \times 100 \%$.}
        \label{fig:motivation:b}
    \end{subfigure}
    \caption{\footnotesize Asymmetric structures along channel and spatial dimension (based on ResNet-18 on ImageNet). Most information in $\tX$ can be embedded into low-rank and low-frequency representations.}
    \label{fig:motivation}
    \vspace{-3mm}
\end{figure}
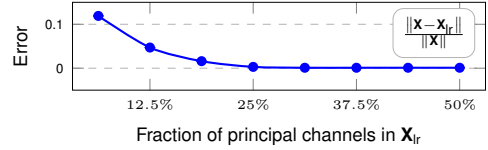
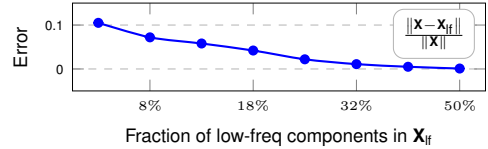

Owing to the spatial correlation, $\tX_{\text{lf}}$ can sufficiently approximate $\tX$ using a few low-frequency components. 
Figure \ref{fig:motivation:b} shows the error \difffro{\tX}{\tX_{\text{lf}}} versus the number of low-frequency components in $\tX_{\text{lf}}$ after the first convolution layer in ResNet-18.
We can also observe that most information in $\tX$ can be embedded into $\tX_{\text{lf}}$ with low-frequency components.

\iffalse
\begin{remark}[Asymmetric Structure in Language Models]
While the asymmetric structure above is observed in CNNs, other model architectures such as Transformers also exhibit a similar asymmetric structure (See Appendix \ref{appx:asymmetricNLP}). 
\end{remark}
\fi

The asymmetric structure also exists in other tasks, such as language models (See Appendix \ref{appx:asymmetricNLP}). 
Observing the asymmetric structures in different dimensions, we aim to design an \emph{asymmetric learning} framework that learns privacy-sensitive low-dimensional IR with a low-complexity model in a private environment, while sending residuals to a larger model trained on an untrusted platform. 

%Hence, training and inference involving privacy-sensitive data are delegated to a private environment without adding substantial complexities. 

\section{\method{}: Private Asymmetric Learning}\label{sec:method}

\textbf{Overview. }
At a high level, as shown in Fig.  \ref{fig:overview}, given a model, \method{} keeps the first few layers in a private environment as a backbone model \MB{} and a main model \MM{}. \method{} then offloads all remaining layers to a public environment as residual model \MRS{}.
Specifically, \method{} uses the backbone model \MB{} as a feature extractor to compute intermediate representations (IRs). \method{} then \emph{asymmetrically} decomposes the IRs, obtaining a low-dimensional information-sensitive part (\IRM{}) and residuals \IRR{}. 
\method{} designs a new low-dimensional model \MM{} for \IRM{} to reduce computation complexity in the private environment. 
On the other hand, the residuals \IRR{} are perturbed and quantized before being outsourced to the untrusted public environment that hosts \MRS{}.
At last, output logits from \MM{} and \MRS{} are added in the private environment, leading to the final predictions. 
These final predictions are not disclosed to the public environment. 
Hence, during training and inference, \method{} allows only minimal residual information to be leaked to the public environment.   

\iffalse
\begin{figure*}[!htb]
    \centering
    \includegraphics[width=.7\linewidth]{fig/SVDDCT.pdf}
    \caption{The asymmetric IR decomposition is shown. We use SVD and DCT to encode channel and spatial information into a low-dimensional representation, and offload the residuals to public environments. }
    \label{fig:IRDecomp}
\end{figure*}
\fi

\begin{figure*}[!htb]
    \centering
    \begin{tikzpicture}
        % private flow
        \node[data, draw=purple!40, very thin, fill=purple!10, text width=5mm, minimum height=6mm](IR1){};
        \node[data, draw=purple!40, very thin, fill=purple!10, text width=5mm, minimum height=6mm, below=-6mm of IR1, xshift=1mm](IR2){};
        \node[data, draw=purple!40, very thin, fill=purple!10, text width=5mm, minimum height=6mm, below=-6mm of IR2, xshift=1mm](IR3){};
        \node[data, draw=purple!40, very thin, fill=purple!10, text width=5mm, minimum height=6mm, below=-6mm of IR3, xshift=1mm](IR4){};
        \node[data, draw=purple!40, very thin, fill=purple!10, text width=5mm, minimum height=6mm, below=-6mm of IR4, xshift=1mm](IR){\scriptsize IR};
        % \draw[draw=red, rounded corners]  (IR1.south west) to node [below] {\scriptsize $c$} (IR.south west);
        \draw [decorate, thin, blue, decoration={brace,amplitude=3pt,mirror}] ($(IR1.south west) + (-.5mm, 0mm)$) -- ($(IR.south west) + (.5mm, 0mm)$) node[midway,yshift=-2mm]{\scriptsize \textcolor{black}{$c$}};

        \node[op, draw=pink!50, fill=pink!30, text width=5mm, right=8mm of IR, minimum height=5mm](SVD){\scriptsize SVD};

        \node[data, draw=purple!40, very thin, fill=purple!10, text width=5mm, minimum height=6mm, right=8mm of SVD](V){};
        \node[data, draw=purple!40, very thin, fill=purple!10, text width=5mm, minimum height=6mm, below=-5mm of V, xshift=1mm](V1){\scriptsize $V^i$};
        \draw[draw=blue]  (V.south west) -- node [below] {\scriptsize $r$} (V1.south west);
        % \draw [decorate, thin, red, decoration={brace,amplitude=2pt,mirror}] ($(V.south west) + (-.5mm, 0mm)$) -- ($(V1.south west) + (.5mm, 0mm)$) node[midway,yshift=-2mm]{\scriptsize $r$};

        \node[op, draw=pink!50, fill=pink!30, text width=5mm, right=8mm of V, minimum height=5mm](DCT){\scriptsize DCT};

        \node[data, draw=purple!40, very thin, fill=purple!10, text width=5mm, minimum height=6mm, right=8mm of DCT](C){};
        \node[data, draw=purple!40, very thin, fill=purple!10, text width=5mm, minimum height=6mm, below=-5mm of C, xshift=1mm](C1){\scriptsize $C^i$};
        \node[data, draw=purple!70, very thin, fill=purple!40, text width=1mm, minimum height=1mm, below=-6mm of C, xshift=-1.9mm](Cs){};
        \node[data, draw=purple!70, very thin, fill=purple!40, text width=1mm, minimum height=1mm, below=-5mm of C, xshift=-1mm](C1s){};
        \draw[draw=blue]  (C.south west) -- node [below] {\scriptsize $r$} (C1.south west);

        \node[op, draw=pink!50, fill=pink!30, text width=5mm, right=8mm of C, minimum height=5mm](IDCT){\scriptsize IDCT};

        \node[data, draw=purple!40, very thin, fill=purple!10, text width=2mm, minimum height=4mm, right=8mm of IDCT](IRM){};
        \node[data, draw=purple!40, very thin, fill=purple!10, text width=2mm, minimum height=4mm, below=-3mm of IRM, xshift=1mm](IRM1){};
        \draw[draw=blue]  (IRM.south west) -- node [below] {\scriptsize $r$} (IRM1.south west);
        \node[right=1mm of IRM1]{\scriptsize \IRM{} low-rank and smaller!};

        \draw[->, red!50, thick] (IR) to (SVD);
        \draw[->, red!50, thick] (SVD) to (V);
        \draw[->, red!50, thick] ($(V.east) + (1mm, 0mm)$) to (DCT);
        \draw[->, red!50, thick] (DCT) to (C);
        \draw[->, red!50, rounded corners] (C1s.north) -- +(0mm, 3mm) -- +(12mm, 3mm) node[above] {\tiny $C^i_{\text{lf}}$} -| ([xshift=-1mm] IDCT.north);
        \draw[->, red!50, rounded corners] (Cs.north) -- +(0mm, 3mm) -- +(14mm, 3mm) -| ([xshift=0.5mm] IDCT.north);
        \draw[->, red!50, thick] (IDCT) to (IRM);

        % public flow
        \node[data, draw=black!40, very thin, fill=black!10, text width=5mm, minimum height=6mm, above=6mm of SVD, xshift=-4mm](IRRSVD1){};
        \node[data, draw=black!40, very thin, fill=black!10, text width=5mm, minimum height=6mm, below=-6mm of IRRSVD1, xshift=1mm](IRRSVD2){};
        \node[data, draw=black!40, very thin, fill=black!10, text width=5mm, minimum height=6mm, below=-6mm of IRRSVD2, xshift=1mm](IRRSVD3){};
        \node[data, draw=black!40, very thin, fill=black!10, text width=5mm, minimum height=6mm, below=-6mm of IRRSVD3, xshift=1mm](IRRSVD4){};
        \node[data, draw=black!40, very thin, fill=black!10, text width=5mm, minimum height=6mm, below=-6mm of IRRSVD4, xshift=1mm](IRRSVD){};
        \node[left=0.1mm of IRRSVD1](){\scriptsize $\tX_{\text{SVD res}}$};

        \node[data, draw=black!40, very thin, fill=black!10, text width=5mm, minimum height=6mm, above=6mm of DCT, xshift=-1mm](IRRDCT1){};
        \node[data, draw=black!40, very thin, fill=black!10, text width=5mm, minimum height=6mm, below=-6mm of IRRDCT1, xshift=1mm](IRRDCT){};
        \node[left=0.1mm of IRRDCT1](){\scriptsize $\tX_{\text{DCT res}}$};

        \node[op, draw=black!50, fill=black!30, text width=5mm, minimum height=5mm, right=26mm of IRRDCT](ADD){\scriptsize Add};

        \node[data, draw=black!40, very thin, fill=black!10, text width=5mm, minimum height=6mm, right=7mm of ADD](IRR1){};
        \node[data, draw=black!40, very thin, fill=black!10, text width=5mm, minimum height=6mm, below=-6mm of IRR1, xshift=1mm](IRR2){};
        \node[data, draw=black!40, very thin, fill=black!10, text width=5mm, minimum height=6mm, below=-6mm of IRR2, xshift=1mm](IRR3){};
        \node[data, draw=black!40, very thin, fill=black!10, text width=5mm, minimum height=6mm, below=-6mm of IRR3, xshift=1mm](IRR4){};
        \node[data, draw=black!40, very thin, fill=black!10, text width=5mm, minimum height=6mm, below=-6mm of IRR4, xshift=1mm](IRR){};
        % \draw[draw=red]  (IRR1.south west) -- node [below] {\scriptsize $c$} (IRR.south west);
        \node[right=1mm of IRR]{\scriptsize \IRR{} little information!};
        \draw [decorate, thin, blue, decoration={brace,amplitude=3pt,mirror}] ($(IRR1.south west) + (-.5mm, 0mm)$) -- ($(IRR.south west) + (.5mm, 0mm)$) node[midway,yshift=-2mm]{\scriptsize \textcolor{black}{$c$}};

        \draw[->, black!50] (SVD) to (IRRSVD);
        \draw[->, black!50] (DCT) to (IRRDCT);
        \draw[->, black!50] (IRRDCT) to (ADD);
        \draw[->, black!50, rounded corners] (IRRSVD1.north) -- +(0mm, 3mm) -- +(62mm, 3mm) -| (ADD.north);
        % \draw[->, black!50, rounded corners] (IRRSVD1.north) to [out=90, in=90] (ADD.north);
        \draw[->, black!50] (ADD) to (IRR1);
        
    \end{tikzpicture}
    \caption{\footnotesize The asymmetric IR decomposition is shown (See Figure \ref{fig:overview} for the whole pipeline). We use SVD and DCT to encode channel and spatial information into a low-dimensional representation, and offload the residuals to public environments. The low-dimensional \IRM{} has fewer channels and smaller sizes but still encode most sensitive information. The residuals \IRR{} have the same dimension as the original IR. }
    \label{fig:IRDecomp}
    \vspace{-4mm}
\end{figure*}
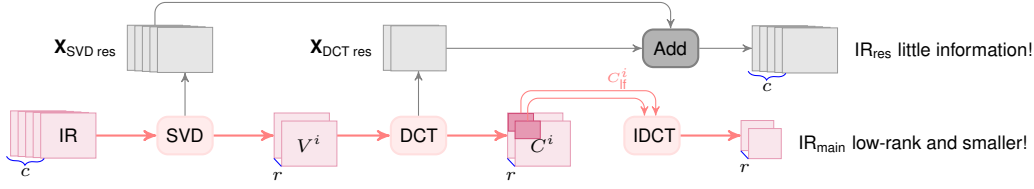

\subsection{Asymmetric IR Decomposition}\label{sec:method:decomp}
As observed in Section \ref{sec:motivation}, IRs after \MB{} exhibit asymmetric structures in multiple dimensions. 
% In this section, we use SVD and DCT to exploit channel and spatial dimension asymmetric structures. 
Hence, we can decompose IRs such that most information is encoded in the low-dimensional IR denoted as \IRM{}. \\
\textbf{Singular Value Decomposition (SVD).} Given an IR $\tX\in\sR^{c\times h\times w}$ obtained from the backbone model \MB{}, we first apply SVD as explained in Sec. \ref{sec:motivation-SVD} to obtain the principal channels $\left \{  V^i \in\sR^{h \times w} \right \}_{i=1}^c$ and the corresponding coefficients $\left \{ \etU^i \in\sR^{c\times 1\times 1} \right \}_{i=1}^c$. 
We then select the $r$ most principal channels as a low-rank representation of $\tX$ as in Sec. \ref{sec:motivation-SVD}.
The rest of the channels are saved as SVD residuals as 
\begin{equation}
    \tX_{\text{SVD res}} = \tX - \tX_{\text{lr}} = \sum_{i=r+1}^c s_i \cdot \etU^i \cdot V^i.
\end{equation} 
\\
\textbf{Discrete Cosine Transform (DCT).} After the decomposition along channels, we further decompose $V^i$ over the spatial dimension using DCT. \\
Specifically, for each principal channel $V^i$, we first obtain the frequency-domain coefficients as $\etC^{i} = \text{DCT}(V^i, T)$.
Then, we only use the low-frequency component to reconstruct a representation as $V^i_{\text{lf}} = \text{IDCT}(\etC^i_{\text{lf}}, T_{\text{lf}})\in\sR^{\frac{h}{t}t'\times \frac{w}{t}t'}$ as in Sec \ref{sec:motivation-DCT}.
$\etC^i_{\text{lf}}$ has a reduced dimension and keeps only top-left low-frequency coefficients in $\etC^i$ (as shown in Figure \ref{fig:IRDecomp}).
$T_{\text{lf}}\in\sR^{t'\times t'}$ corresponds to DCT transformation matrix with reduced spatial dimension. 
The rest of the high-frequency components are saved as DCT residuals as
\begin{equation}
    V^i_{\text{DCT res}} = \text{IDCT}(\etC^i_{\text{res}}, T),
\end{equation}

\noindent where $\etC^i_{\text{res}}$ denotes $\etC^i$ with zeros on the top-left corner.

After SVD and DCT, we obtain the privacy-sensitive low-dimensional features as
\begin{equation}\label{eq:IRMain}
    \text{IR}_{\text{main}} = \sum_{i=1}^r s_i \cdot \etU^i \cdot V^i_{\text{lf}}.
\end{equation}

\noindent On the other hand, the residuals to be offloaded to the untrusted public environment are given as

\begin{equation}
\begin{split}
    &\text{IR}_{\text{res}}  = \tX - \text{IR}_{\text{main}} = \tX_{\text{SVD res}} + \tX_{\text{DCT res}}  \\
    & = \sum_{i=r+1}^c s_i \cdot \etU^i \cdot V^i \quad + \quad \sum_{i=1}^r s_i \cdot \etU^i \cdot V^i_{\text{DCT res}}.
\end{split}
\end{equation}

\noindent \IRR{} is further normalized as
\begin{equation}\label{eq:IRResNorm}
    \text{IR}_{\text{res,norm}} = \text{IR}_{\text{res}} / \max( 1, \left \| \text{IR}_{\text{res}} \right \|_2/C ),
\end{equation}
where $C$ is a scaling parameter for $\ell_2$ normalization. The normalization is necessary to bound \emph{sensitivity} for DP. For simplification, we denote the normalized residuals as \IRR{}.

Hence, \IRM{} has fewer principal channels and smaller spatial dimensions but contains most information in IR. 
The non-principal channels and high-frequency components, on the other hand, are saved in \IRR{}.
\IRM{} and \IRR{} are then respectively fed into a small model \MM{} in a private environment and a large model \MRS{} in a public environment.

\subsection{Residuals Perturbation and Quantization}\label{sec:method:quant}
While the low-dimensional representation \IRM{} has the most important information in the IR, \IRR{} might still contain some information such as a few high-frequency components.
Furthermore, as the communication from a private environment (e.g., TEEs) to a public environment (e.g., GPUs) is usually slow, sending floating-point high-dimensional residuals can significantly increase the communication overhead and prolong the total running time.

In this section, we perturb \IRR{} with a Gaussian mechanism and then apply binary quantization on the perturbed \IRR{} to reduce the inter-environment communication cost. \\
Given a DP budget $\epsilon$, we consider the Gaussian mechanism, and compute the corresponding noise parameter $\sigma$ (See Section \ref{sec:analysis}). 
For each tensor \IRR{}, we generate a noise tensor $\tN\in\sR^{c\times h \times w} \sim \mathcal{N}(0, \sigma^2\cdot \tI)$, and add it to \IRR{} in the private environment.
With noisy residuals, \IRN{}=\IRR{} $+ \ \tN$, we apply a binary quantizer as follows
\begin{equation}\label{eq:binquant}
    \small \text{IR}_{\text{quant}} = \mathrm{\small BinQuant}( \text{IR}_{\text{noisy}} ) =
    \begin{cases}
        0 \quad \text{IR}_{\text{noisy}} < 0, \\
        1 \quad \text{IR}_{\text{noisy}} \geq 0.
    \end{cases}
\end{equation}

As a result, the tensor to be offloaded to the public environment is a binary representation of the residuals. 
Compared to floating-point values, such a binary representation reduces communication by $32\times$.
Owing to the asymmetric decomposition, the values in \IRR{} are usually close to zero. Hence, a small noise is sufficient to achieve strong privacy protection (See formal analysis in Section \ref{sec:analysis}). 
Further ablation studies can be found in Section \ref{appx:ablation:quantization}.
% Furthermore, as the residuals usually contain high-frequency components (e.g., noise, the outline of objects), binary representation after $\mathrm{BinQuant}$ can still preserve the pattern without introducing significant distortion. 
% In addition, as \MM{}'s logits vector is propagated to \MRS{} during backward passes in training, we add a small perturbation noise before sending it to \MRS{}. 
% Since the logits vector is a normalized vector, we can easily compute the noise required according to the standard Laplacian mechanism \footnote{Alternatively, \MRS{} can be trained even without \MM{}'s logits. We discuss this scheme in Appendix \ref{appx:subsec:logits}.}.

\subsection{Model Design for Low-Dimensional \IRM{}}\label{sec:method:model}
Knowing that \IRM{} has a low-rank as given in Equation (\ref{eq:IRMain}), in this section, we show that developing an efficient \MM{} with low computation complexity is attainable.
% We target linear layers in neural nets and design a more efficient alternative given a low-dimension input \IRM{}. 

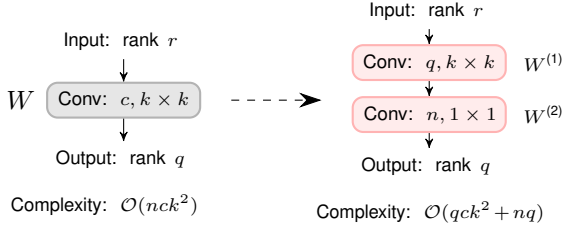
\begin{figure}[!htb]
    \centering
    \begin{tikzpicture}

        % standard model
        \node[op, draw=black!30, fill=black!10, text width=18mm](CONV){\scriptsize Conv: $c, k\times k$};
        \node[text width=16mm, above=3mm of CONV](IN){\scriptsize Input: rank $r$};
        \node[text width=18mm, below=3mm of CONV](OUT){\scriptsize Output: rank $q$};
        \node[text width=30mm, below=1mm of OUT](){\scriptsize Complexity: $\mathcal{O}(nck^2)$};
        \node[text width=3mm, left=1mm of CONV](W){$W$};

        \draw[->] (IN) -- (CONV);
        \draw[->] (CONV) -- (OUT);

        % low-rank model
        \node[op, draw=red!30, fill=pink!30, text width=18mm, right=20mm of CONV, yshift=5mm](CONV1){\scriptsize Conv: $q, k\times k$};
        \node[op, draw=red!30, fill=pink!30, text width=18mm, below=2mm of CONV1](CONV2){\scriptsize Conv: $n, 1\times 1$};
        \node[text width=16mm, above=2mm of CONV1](IN1){\scriptsize Input: rank $r$};
        \node[text width=18mm, below=2mm of CONV2](OUT1){\scriptsize Output: rank $q$};
        \node[text width=30mm, below=1mm of OUT1](){\scriptsize Complexity: $\mathcal{O}(qck^2+nq)$};
        \node[text width=3mm, right=1mm of CONV1](V){\scriptsize $W^{\text{(1)}}$};
        \node[text width=3mm, right=1mm of CONV2](U){\scriptsize $W^{\text{(2)}}$};

        \draw[->] (IN1) -- (CONV1);
        \draw[->] (CONV1) -- (CONV2);
        \draw[->] (CONV2) -- (OUT1);

        % relation
        \draw [-{Stealth[length=3mm, width=2mm]}, dashed] ([xshift=4mm] CONV.east) -- +(12mm, 0mm);
        % \node[draw=black!50, rounded corners, text width=22mm, below=10mm of CONV, align=center](DES){original conv};
        % \node[draw=black!50, rounded corners, text width=23mm, right=17mm of DES, align=center](DES1){low-rank conv};
    \end{tikzpicture}
    \caption{\footnotesize Model design for the low-dimensional \IRM{}. Knowing the rank in data, the number of channels in convolution layers can be reduced, leading to a reduction in computation complexity.}
    \label{fig:model}
    \vspace{-4mm}
\end{figure}

\iffalse
\begin{figure}[!htb]
    \centering
    \includegraphics[width=0.9\linewidth]{fig/Model.pdf}
    \caption{Model design for the low-dimensional \IRM{}.}
    \label{fig:model}
\end{figure}
\fi
In linear layers such as convolutional layer, low-rank inputs lead to low-rank outputs \cite{legrace_PETS_2022}. 
% As shown in \cite{legrace_PETS_2022}, given low-rank inputs with $r$ principal channels, the outputs after linear layers also have a low-rank structure with $q$ channels. For instance, if kernel size $k=3\times 3$, $q$ will be $2r$. 
Assuming inputs and outputs have rank $r$ and $q$, we split the convolution layer into two sub-layers as shown in Figure \ref{fig:model} (right). The first layer has $q$  ($k\times k$) kernels to learn the principal features, whereas the second layer has $c$ ($1\times 1$) kernels to combine the principal channels. 
We further add a kernel orthogonality regularization \cite{Orth_CVPR_2017} to the first layer to enhance the orthogonality of the output channels. 

\iffalse
Such an adaptation comes with prior knowledge of the ranks of inputs and outputs. For an original convolution layer with kernel $\tW\in\sR^{n\times c\times k\times k}$, given input $\tX\in\sR^{c\times h\times w}$, output is calculated as $\tY = \tW \circledast \tX$.
We can also write $\tY$ in a matrix form as\footnote{See Appendix \ref{} for further details} $Y = W \cdot X$, where $W\in\sR^{n\times ck^2}$, and $X\in\sR^{ck^2\times hw}$.
On the other hand, we can calculate the output for the new convolution layer as $\tY^{'} = \tW^1 \circledast \tW^2 \circledast \tX$.
Similarly, a matrix form can be written as $Y^{'} = W^2 \times W^1 \times X$, where $W^2\in\sR^{n\times q}$ and $W^1\in\sR^{q\times ck^2}$.
Given $Rank(Y) = q$, there exists an solution $W^1, W^2$ such that \norm{Y - Y^{'}} = 0, namely $Y^{'}$ can be the exact low-rank presentation of $Y$.
Therefore, with low-rank inputs, we can design a low-rank model as in Figure \ref{fig:model} to preserve information as in the original layer.
\fi
Such a design comes with fundamental reasoning.
Theorem \ref{th:model} shows that, with knowledge of the ranks of inputs and outputs in a convolutional layer, it is possible to optimize a low-dimensional layer as in Figure \ref{fig:model} (right) such that it results in the same output as the original layer.
\begin{theorem}\label{th:model}
    For a convolution layer with weight $\tW\in\sR^{n\times c\times k\times k}$ with an input $\tX$ with rank $r$ and output $\tY$ with rank $q$, there exists an optimal $\tW^{(1)}\in\sR^{q\times c\times k\times k}, \tW^{(2)}\in\sR^{n\times q\times 1\times 1}$ in the low-dimensional layer such that the output of this layer denoted as $\tY^{'}$ satisfies
\begin{align}
    \min_{\tW^{(1)}, \tW^{(2)}}\left \| \tY - \tY^{'}  \right \| = 0.
\end{align}
\end{theorem}
\noindent The proof is deferred to Appendix \ref{appx:model}.

\begin{remark}
    While the low-dimensional layer in Figure \ref{fig:model} (right) shares similar architecture as low-rank compression methods \cite{Lowrank_BMVC_2014,Lowrank_ICLR_2016, SensiSVD_GlobalSIP_2017}, low-rank compression inevitably incurs information loss in outputs. However, the layer in Figure \ref{fig:model} (right) can theoretically preserve all information given low-rank inputs and outputs.
\end{remark}

\subsection{Private Backpropogation}\label{sec:method:backprop}
While the asymmetric IR decomposition together with the randomized quantization mechanism ensures privacy in forward passes, information of \IRM{} can still be leaked through backpropagation on logits. 
In this section, we further propose a private backpropagation that removes the logits of \MM{} from the gradients to \MRS{}.

In detail, the gradients to \MM{} are calculated through a standard backpropagation algorithm that uses logits from both \MM{} and \MRS{}. While gradients to \MRS{} are calculated solely using \MRS{} logits. 
Hence, the logits from \MM{} will not be revealed to the outside. 
Specifically, given the output logits from \MM{} and \MRS{}: \zM{} and \zRS{}, we compute \texttt{Softmax} for $i$-th label \MM{}  and \MRS{} as
\begin{equation*}
\begin{split}
    \mathcal{M}_{\text{main}}: \bm{o}_{\text{tot}}(i) &= {e^{ \bm{z}_{\text{main}}(i) + \bm{z}_{\text{res}}(i) }} / {\sum\nolimits_{j=1}^{L} e^{ \bm{z}_{\text{main}}(j) + \bm{z}_{\text{res}}(j) } }, \\
    \mathcal{M}_{\text{res}}: \bm{o}_{\text{res}}(i) &= {e^{ \bm{z}_{\text{res}}(i) }} / {\sum\nolimits_{j=1}^{L} e^{ \bm{z}_{\text{res}}(j) } },
\end{split}
\end{equation*}
where $L$ denotes the number of labels in the current task. Following the backpropagation in \texttt{Softmax}, we compute gradients to \MM{} and \MRS{}: $\bm{g}_{\text{main}}$, $\bm{g}_{\text{res}}$, as 

\vspace{-3mm}

\begin{equation*}
    \bm{g}_{\text{main}} = \bm{o}_{\text{tot}} - \bm{y}, \quad \bm{g}_{\text{res}} = \bm{o}_{\text{res}} - \bm{y},
\end{equation*}
where $\bm{y}$ denotes one-hot encoding for labels. 

With the separate backpropagation on gradients, \MM{} avoids revealing its logit to \MRS{}, while still using \MRS{}'s logits for its own backpropagation.  

\subsection{Training Procedure}\label{sec:method:training}
The model training using \method{} consists of two stages.
\\
\textbf{Stage 1. } We train the backbone, \MB{}, and the main model, \MM{}, with \IRM{} only. \IRR{} is ignored and not shared with the public. Therefore, there is no privacy leakage at this stage, as all data and model parameters are kept in the private environment. After stage 1, we cache all residual data \IRR{} from SVD and DCT decomposition, and apply the randomized quantization mechanism once on \IRR{}. 
\\
\textbf{Stage 2. } We freeze the backbone model, \MB{}, and continue to train the main model, \MM{}, and the residual model, \MRS{}. As \MB{} is frozen, we directly sample residual inputs for \MRS{} from the cached residual data. While for \MB{} and \MM{}, we fetch data from the raw datasets, apply SVD and DCT decomposition, and send \IRM{} to \MM{}.

\noindent The training procedure is provided in algorithm \ref{alg:method}.

\begin{algorithm}[!htb]
\caption{\method{} training procedure}
\label{alg:method}
\begin{algorithmic}[1]
\REQUIRE $\mathrm{ep}_1$: \#epochs at stage 1, $\mathrm{ep}_2$: \#epochs at stage 2 \\
\REQUIRE $\epsilon$: privacy constraint \\
\STATE Initialize \MB{}, \MM{}, \MRS{} \\
\FOR{$t = 1, \cdots, \mathrm{ep}_1$}    
    \FOR{ a batch \blue{from} the dataset }
    \STATE Train \MB{} and \MM{}. \\
\ENDFOR
\ENDFOR
\STATE Freeze \MB{} and cache all residual data \IRR{}.
\STATE Apply the randomized quantization based on Eq (\ref{eq:binquant}). \\
\FOR{$t = 1, \cdots, \mathrm{ep}_2$}    
    \FOR{ a batch \blue{from} the dataset, cached residuals }
    \STATE Train \MM{} and \MRS{}. \\
    \ENDFOR
\ENDFOR
\end{algorithmic}
\end{algorithm}

\section{Privacy Analysis}\label{sec:analysis}

This section provides the differential privacy analysis for \method{}. 
As the backbone model \MB{}, the low-dimensional model \MM{} and final predictions remain private, the only public information are the residuals and the residual model \MRS{} at stage 2 during training. 
Therefore, we analyze the privacy leakage of the perturbed residuals at stage 2. 

Given neighboring datasets $\mathcal{D} = \left \{ X^1, \cdot, X^i, \cdot, X^N \right \}$ and $\mathcal{D}^{'} = \left \{ X^1, \cdot, 0, \cdot, X^N \right \}$ with $i$-th record removed, the global $\ell_2$-sensitivity is defined as 
$\Delta_2 = \sup_{\tX_{\text{res}}, \tX^{'}_{\text{res}}} \left \| \tX_{\text{res}} - \tX^{'}_{\text{res}} \right \|_F \leq C$,
where $\tX_{\text{res}}, \tX^{'}_{\text{res}}$ denote residuals obtained from $\mathcal{D}$ and $\mathcal{D}^{'}$. 
\\We provide our DP guarantee in Theorem \ref{th:privacy}.

\begin{theorem}\label{th:privacy}
\method{} ensures that the perturbed residuals 
%in Equation (\ref{eq:binquant})
and operations in the public environment satisfy $(\epsilon,\delta)$-DP given noise $\tN \sim \mathcal N (0, 2C^2\cdot\log{(2/\delta') / \epsilon'} )$ given sampling probability $p$, and $\epsilon=\log{(1 + p(e^{\epsilon'}-1))}, \delta=p\delta'$.
\end{theorem}

%\begin{proof}
\noindent The proof of Theorem \ref{th:privacy} relies on the analysis of Gaussian mechanism in and the post-processing rule of DP \cite{AproximatedDP,DPFoundation}, and is provided in Appendix \ref{appx:privacy} for completeness.
%\end{proof}

\begin{remark}
    (Training with Cached Residual Improves DP). Sampling an image multiple times does not incur additional privacy leakage. As described in Sec. \ref{sec:method:training}, the perturbation is only performed once before training the residual model \MRS{}.
    When training the residual model \MRS{}, \method{} directly samples input for \MRS{} from the perturbed cached residuals, with no need to perform perturbation on the fly \cite{GAP}.
\end{remark}

\section{Empirical Evaluations}\label{sec:eval}

In this section, we evaluate \method{} in terms of model accuracy, running time, and resilience against attacks. We conduct further ablation studies on different ways of merging logits and effects of perturbation in Appendix \ref{appx:ablation}.\\
\textbf{Datasets and Models.} We use CIFAR10/100 \cite{CIFAR} and ImageNet \cite{ImageNet}. For models, we choose ResNet-18 and ResNet-34 \cite{ResNet_2018_CVPR}. %We follow standard data processing ops during training. 
Hyperparameters are provided in Appendix \ref{appx:exp:hparam} \\
\textbf{Model Configuration.} For ResNet models, \MB{} consists of the first convolution layer, while all \emph{ResBlocks} \cite{ResNet_2018_CVPR} and the fully-connected layer are offloaded in \MRS{}. On the other hand, \MM{}'s details are deferred to Appendix \ref{appx:arch}. \\
\textbf{Asymmetric Decomposition}. For the SVD, based on \cite{legrace_PETS_2022}, we keep $r=8$ principal channels in \IRM{} for CIFAR-10/100 and $12$ for ImageNet to keep $>95\%$ information in the private environment. For the DCT, we set $t,t'=16, 8$ on CIFAR, and $14, 7$ on ImageNet to avoid noticeable information leakage in the residuals.
Note that while decomposition with larger $r$ and $t'$ leads to less information leakage, it also incurs more computations in the private environment. 
Therefore, we choose the parameters to maintain a reasonable trade-off between privacy leakage and complexity. 

% \textbf{Privacy Parameters}. During training we use a batch size $b=64$ and $256$ for CIFAR-10/100 and ImageNet. We normalize \IRR{} and adjust $C$ in Equation (\ref{eq:IRResNorm}) to ensure that $p\cdot\Delta=1$ in Theorem \ref{th:privacy}. \\
% \textbf{Training Procedure.} At the first stage, we train \MB{} and \MM{} for one third of the total number of epochs. Then, at the second stage, we freeze the backbone \MB{} and continue training \MM{} and \MRS{}. As \MB{} is frozen at the second stage, the output residual \IRR{} from \MB{} is fixed given an input. Therefore, we apply the perturbation (Section \ref{sec:method:quant}) once on \IRR{} throughout the rest of the training.

\iffalse
\subsection{Model Accuracy}
\begin{itemize}
    \item Basic: Accuracy (Original model, M1, M1+M2 \blue{with and without exachanging logits with M1, with and without DCT, compare with Asymml}) on various datasets and models [\red{Done}].

    \item Extend: Experiments on large models such as ResNet-34 [\red{Running}].
\end{itemize}
\fi

\subsection{Model Accuracy}
In this section, we evaluate model accuracy. 
We train ResNet-18 on CIFAR-10/100, and ResNet-18/34 on ImageNet.
We first compare the accuracy of the following schemes.
\begin{itemize}
    \item \MM{}: Train \MB{}+\MM{} with no residuals.
    \item \MM{}+\MRS{} : Train \MB{} + (\MM{}, \MRS{}) with DP.
    \item Orig: Train the original model without \method{}. 
\end{itemize}

Figure \ref{fig:acc:cifar} and \ref{fig:acc:imagenet} show the final accuracy on CIFAR-10/100 and ImageNet.
First, owing to the effective asymmetric IR decomposition and the low-dimensional model, \MM{} already gives an accuracy that is close to the original model on both CIFAR-10/100 and ImageNet datasets. 
With adding \MRS{}, \method{} achieves a comparable accuracy as the original model. 
With the residual information further protected by Gaussian noise, we observe that \method{} strikes a much improved privacy-utility trade-off, with slight performance degradation under a low privacy budget (small $\epsilon$).

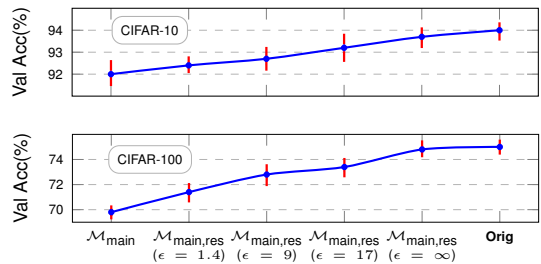
\begin{figure}[!htb]
    \centering
    \begin{subfigure}{.45\textwidth}
        \centering
        \hspace{-4mm}
        \begin{tikzpicture}
            \begin{axis}[
            width=.99\linewidth, height=.35\linewidth,
            xmin=-0.5, xmax=5.5,
            xtick={0, 1, 2, 3, 4, 5},
            %xticklabels={
            %    $\mathcal{M}_{\text{main}}$, 
            %    {$\mathcal{M}_{\text{main,res}}$ \\ ($\epsilon=1$)}, 
            %    {$\mathcal{M}_{\text{main,res}}$ \\ ($\epsilon=2$)}, 
            %    {$\mathcal{M}_{\text{main,res}}$ \\ ($\epsilon=3$)}, 
            %    {$\mathcal{M}_{\text{main,res}}$ \\ ($\epsilon=\infty$)}, 
            %    \textbf{Orig} 
            %},
            % xticklabel style={ align=center,text width=14mm },
            xticklabels={,,},
            ymin=91, ymax=95,
            ylabel={\scriptsize Val Acc(\%)},
            ytick={92, 93, 94},
            yticklabels={92, 93, 94},
            ymajorgrids,
            major grid style={line width=.2pt, draw=black!30, dashed},
            % extra y ticks={92, 94},
            % extra y tick labels={92, 94},
            % extra tick style={major grid style=red, dashed},
            ticklabel style={font=\tiny},
            legend style={at={(1.0,0.35)}},
            ]
    
            % with freezing
            \addplot+[ 
                thick, blue, smooth, mark=*, mark options={solid, scale=0.4},
                error bars/.cd, y fixed, y dir=both, y explicit, error mark options={draw=red, line width=.3mm}
            ] 
            table [x=x, y=y, y error plus=error1, y error minus=error2, col sep=comma]{
                x,    y,      error1,    error2
                % (91.757, 92.358, 91.937)
                0,    92.0,   0.4,       0.3
                % (92.428, 92.398, 92.498)
                1,    92.4,   0.17,      0.11
                % (92.778, 92.628, 92.859)
                2,    92.7,   0.3,       0.3
                % (93.229, 92.768, 93.479)
                3,    93.2,   0.4,       0.4
                % ()
                4,    93.7,   0.19,     0.27
                5,    94,     0.12,     0.23
            };
            % \addlegendentry{\tiny Delta \MB freeze}    

            \node[fill=white, draw=black!30, rounded corners] at (axis cs:0.5,94){\tiny CIFAR-10};
            % \node[fill=white, draw=black!30, rounded corners] at (axis cs:1,94){\tiny Original Model};
            \end{axis}
        \end{tikzpicture}
        %\caption{Validation accuracy on CIFAR-10}
        %\label{fig:acc:cifar10}
    \end{subfigure}

    \vspace{1mm}
    
    \begin{subfigure}{.45\textwidth}
        \centering
        \begin{tikzpicture}
        \begin{axis}[
            xmin = -0.5, xmax = 5.5, ymin = 69, ymax = 76,
            ylabel = {\scriptsize Val Acc(\%)}, 
            xtick = { 0, 1, 2, 3, 4, 5 },
            xticklabels = { \tiny $\mathcal{M}_{\text{main}}$, 
                            \tiny {$\mathcal{M}_{\text{main,res}}$ \\ ($\epsilon=1.4$)}, 
                            \tiny {$\mathcal{M}_{\text{main,res}}$ \\ ($\epsilon=9$)},
                            \tiny {$\mathcal{M}_{\text{main,res}}$ \\ ($\epsilon=17$)},
                            \tiny {$\mathcal{M}_{\text{main,res}}$ \\ ($\epsilon=\infty$)}, 
                            \tiny \textbf{Orig}},
            xticklabel style={ align=center,text width=14mm },
            ytick = { 70, 72, 74}, yticklabels = { 70, 72, 74 },
            ymajorgrids, major grid style={line width=.1pt, draw=black!30, dashed},
            width=.99\textwidth, height=.35\textwidth, line width=0.1pt,
            ticklabel style={font=\tiny},
            ]
    
            \addplot+[ 
                thick, blue, smooth, mark=*, mark options={solid, scale=0.4},
                error bars/.cd, y fixed, y dir=both, y explicit, error mark options={draw=red, line width=.3mm}
            ] 
            table [x=x, y=y, y error plus=error1, y error minus=error2, col sep=comma]{
                x,    y,      error1,    error2
                0,    69.8,   0.12,      0.15
                % (71.364, 70.993, 71.623)
                1,    71.4,   0.3,      0.4
                % (72.85, 73.187, 72.396)
                2,    72.8,   0.4,       0.5
                % (73.4, 73.12, 73.631)
                3,    73.4,   0.3,       0.4
                % (74.80, 74.57, 75.21)
                4,    74.8,   0.3,      0.2
                5,    75,     0.16,     0.19
            };

            \node[fill=white, draw=black!30, rounded corners] at (axis cs:0.5,74){\tiny CIFAR-100};
        \end{axis}
        \end{tikzpicture}
        % \caption{Validation accuracy on CIFAR-100}
        % \label{fig:acc:cifar100}
    \end{subfigure}
    \caption{\footnotesize Val acc of ResNet-18 on CIFAR-10, CIFAR-100. \MM{} gives accuracy close to the original model. With adding \MRS{}, \method{} achieves comparable accuracy as the original model. By adding noise to \IRR{}, \method{} achieves strong DP while still preserving the model performance.}
    \label{fig:acc:cifar}
    \vspace{-4mm}
\end{figure}

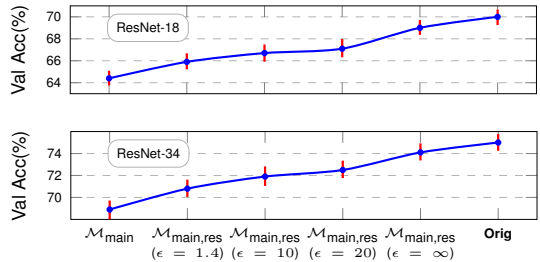
\begin{figure}[!htb]
    \centering
    \begin{subfigure}{.45\textwidth}
        \centering
        \hspace{-4mm}
        \begin{tikzpicture}
        \begin{axis}[
            xmin = -0.5, xmax = 5.5, ymin = 63, ymax = 71,
            ylabel = {\scriptsize Val Acc(\%)}, 
            xtick = { 0, 1, 2, 3, 4, 5 },
            %xticklabels = { \tiny $\mathcal{M}_{\text{main}}$, 
            %                \tiny {$\mathcal{M}_{\text{main,res}}$ \\ ($\epsilon=1$)}, 
            %                \tiny {$\mathcal{M}_{\text{main,res}}$ \\ ($\epsilon=2$)},
            %                \tiny {$\mathcal{M}_{\text{main,res}}$ \\ ($\epsilon=3$)},
            %                \tiny {$\mathcal{M}_{\text{main,res}}$ \\ ($\epsilon=\infty$)}, 
            %                \tiny \textbf{Orig}},
            %xticklabel style={ align=center,text width=14mm },
            xticklabels={,,},
            ytick = {64, 66, 68, 70 }, yticklabels = {64, 66, 68, 70 },
            ymajorgrids, major grid style={line width=.1pt, draw=black!30, dashed},
            width=.99\textwidth, height=.35\textwidth, line width=0.1pt,
            ticklabel style={font=\tiny},
            ]
    
            \addplot+[ 
                thick, blue, smooth, mark=*, mark options={solid, scale=0.4},
                error bars/.cd, y fixed, y dir=both, y explicit, error mark options={draw=red, line width=.3mm}
            ] 
            table [x=x, y=y, y error plus=error1, y error minus=error2, col sep=comma]{
                x,    y,    error1,    error2
                0,    64.4, 0.21,     0.16
                % (65.91, 65.73, 66.23)
                1,    65.9,  0.30,    0.2
                % (66.7, 66.931, 66.524)
                2,    66.7,  0.3,     0.3
                % (67.11, 66.82, 67.45)
                3,    67.1,  0.4,    0.3
                4,    69.0,  0.24,    0.16
                5,    70,    0.18,    0.26
            };

            \node[fill=white, draw=black!30, rounded corners] at (axis cs:0.5,69){\tiny ResNet-18};
        \end{axis}
        \end{tikzpicture}
        % \caption{Validation accuracy on ResNet-18}
        % \label{fig:acc:imagenet:resnet18}
    \end{subfigure}

    \vspace{1mm}
    
    \begin{subfigure}{.45\textwidth}
        \centering
        \begin{tikzpicture}
        \begin{axis}[
            xmin = -0.5, xmax = 5.5, ymin = 68, ymax = 76,
            ylabel = {\scriptsize Val Acc(\%)}, 
            xtick = { 0, 1, 2, 3, 4, 5 },
            xticklabels = { \tiny $\mathcal{M}_{\text{main}}$, 
                            \tiny {$\mathcal{M}_{\text{main,res}}$ \\ ($\epsilon=1.4$)}, 
                            \tiny {$\mathcal{M}_{\text{main,res}}$ \\ ($\epsilon=10$)}, 
                            \tiny {$\mathcal{M}_{\text{main,res}}$ \\ ($\epsilon=20$)}, 
                            \tiny {$\mathcal{M}_{\text{main,res}}$ \\ ($\epsilon=\infty$)}, 
                            \tiny \textbf{Orig}},
            xticklabel style={ align=center,text width=14mm },
            ytick = { 70, 72, 74 }, yticklabels = { 70, 72, 74 },
            ymajorgrids, major grid style={line width=.1pt, draw=black!30, dashed},
            width=.99\textwidth, height=.35\textwidth, line width=0.1pt,
            ticklabel style={font=\tiny},
            ]
    
            \addplot+[ 
            thick, blue, smooth,
            mark=*, mark options={solid, scale=0.4},
            error bars/.cd,
              y fixed,
              y dir=both,
              y explicit, error mark options={draw=red, line width=.3mm}
            ] 
            table [x=x, y=y, y error plus=error1, y error minus=error2, col sep=comma]{
                x,      y,      error1,     error2
                0,      68.9,   0.34,       0.41
                1,      70.8,   0.34,       0.27
                2,      71.9,   0.44,       0.37
                3,      72.5,   0.37,       0.25
                4,      74.1,   0.34,       0.25
                5,      75,     0.29,       0.27
            }; 

            \node[fill=white, draw=black!30, rounded corners] at (axis cs:0.5,74){\tiny ResNet-34};
        \end{axis}
        \end{tikzpicture}
        %\caption{Validation accuracy on ResNet-34}
        %\label{fig:acc:imagenet:resnet34}
    \end{subfigure}
    \caption{\footnotesize Val acc of ResNet-18 and ResNet-34 on ImageNet.}
    \vspace{-3mm}
    \label{fig:acc:imagenet}
\end{figure}

To further study the privacy-utility trade-off, we compare \method{} with a scheme that directly adds noise to IR of the original model (naive-DP) under the same privacy guarantee. As in Equation (\ref{eq:IRResNorm}), we perform normalization on IR before perturbation. 
As shown in Table \ref{tab:acc:dp}, training with noise added to IR incurs significant performance degradation given the same $\epsilon$. 
With the asymmetric IR decomposition, \method{} offers a much better privacy-utility trade-off. 

\iffalse
\begin{table*}[!htb]
    \centering
    \small
    \caption{Model accuracy of ResNet-18 under the same privacy budget ($\epsilon=1.4$). Compared to \method{}, adding noise directly to IR incurs significant performance drops. }
    \label{tab:acc:dp}
    \begin{tabular}{c|cc|cc|cc}
        \toprule
                     & \multicolumn{2}{c|}{CIFAR-10} & \multicolumn{2}{c|}{CIFAR-100} & \multicolumn{2}{c}{ImageNet} \\
        \midrule
        Add noise to  & \IRR{} & IR & \IRR{} & IR & \IRR{} & IR  \\
        Accuracy (\%)  & 92.4 & 69.6$(\downarrow \red{-22.8})$ & 70.9 & 48.3 ($\downarrow \red{-22.6}$) & 65.4 & 34.4 ($\downarrow \red{-31}$) \\
       \bottomrule
    \end{tabular}
\end{table*}
\fi

\begin{table}[!htb]
    \centering
    \small
    \caption{Model accuracy of ResNet-18 by perturbing IR and \IRR{} under the same privacy budget ($\epsilon=1.4$). }
    \label{tab:acc:dp}
    \begin{tabular}{c|cc}
        \toprule
                    Dataset & \method{}: perturb \IRR{} & naive-DP: perturb IR  \\
        \midrule
        CIFAR-10  & $92.4 \%$ & $69.6\%$ $(\downarrow \red{-22.8})$ \\
        CIFAR-100 & $71.4\%$ & $48.3\%$ ($\downarrow \red{-23.1}$) \\
        ImageNet & $65.9\%$ & $34.4\%$ ($\downarrow \red{-31.5}$) \\
       \bottomrule
    \end{tabular}
    \vspace{-3mm}
\end{table}

\subsection{Running Time Analysis} 
In this section, we compare running time using \method{} with a private training framework \texttt{3LegRace} \cite{legrace_PETS_2022} and private inference framework \texttt{Slalom} \cite{Slalom_ICLR_20}.
As a reference, we also include the time of running the original model solely in a private environment (\texttt{Priv-only}).
% We adopt Intel SGX \cite{SGX_2016} as the private environment, while Nvidia RTX 5000 as the public environment. 

\textbf{Theoretical Complexity. } 
Table \ref{tab:model:complexity} lists the theoretical computational complexities by MACs (i.e., multiply-accumulate) of ResNet-18 on CIFAR-10/100, and ResNet-34 on ImageNet during the inference phase with a batch of size $1$.
For SVD, we use an approximation algorithm \cite{legrace_PETS_2022} that only computes the first $r$ principal channels.
First, we can see that SVD and DCT only account for a very small fraction of the total computations, which aligns with the real running time in Table \ref{tab:runtimebreak}.
On the other hand, compared to MACs in \MRS{}, the computation complexity of \MB{}+\MM{} is much smaller, only accounting for $\sim 10\%$ of MACs in \MRS{}.
This shows that, with the asymmetric decomposition that embeds most information into a low-dimensional representation, the low-dimensional model effectively reduces the computation cost of the resource-constrained private environments.

\iffalse
\begin{table*}[!htb]
    \centering
    \caption{MACs of ResNet-18 (CIFAR) and ResNet-34 (ImageNet) during inference with batch size 1. SVD and DCT only incurs very small additional overhead. Comparing to the residual model, the complexity of the backbone and the main model is much smaller. }
    \label{tab:model:complexity}
    \begin{tabular}{cccc|cccc}
    \toprule
         \multicolumn{4}{c|}{ResNet-18 (CIFAR)} & \multicolumn{4}{c}{ResNet-34 (ImageNet)} \\ \midrule
        \MB{}+\MM{} & SVD & DCT & \MRS{} & \MB{}+\MM{} & SVD & DCT & \MRS{}{} \\ \midrule
        48.3 M & 0.52 M & 0.26 M & 547M & 437 M & 1.6 M & 0.7 M & 3.5G \\
    \bottomrule
    \end{tabular}
\end{table*}
\fi

\begin{table}[!htb]
    \centering
    \small
    \caption{\footnotesize MACs of ResNet-18 and ResNet-34 during forward passes with batch size 1. Compared to the residual model, the complexity of the backbone and the main model is much smaller. }
    \label{tab:model:complexity}
    \begin{tabular}{c|cccc}
    \toprule
        Model & \MB{}+\MM{} & SVD & DCT & \MRS{} \\ \midrule
        ResNet-18 & 48.3 M & 0.52 M & 0.26 M & 547M  \\
        ResNet-34 & 437 M & 1.6 M & 0.7 M & 3.5G \\
    \bottomrule
    \end{tabular}
    \vspace{-2mm}
\end{table}

\textbf{Real Running Time. }
We test ResNet-18 on CIFAR-100 with batch size $32$ and average per-iteration time across one epoch.
We use Intel SGX \cite{SGX_2016} as a private environment, and Nvidia RTX 5000 as an untrusted public environment.

\iffalse
\begin{table*}[!htb]
    \caption{Training and inference time of \method{} and other baselines (ResNet-18/CIFAR-100, $b=32$). \method{} achieves significant speedup compared to \texttt{3LegRace} in training, and \texttt{Slalom} in inference. The speedup comes from reduced layer-wise TEE-GPU communication in both \texttt{3LegRace} and \texttt{Slalom}.}
    \label{tab:runtime}
    \centering
    \begin{tabular}{c|ccc|ccc}
        \toprule
          & \multicolumn{3}{c|}{Training} & \multicolumn{3}{c}{Inference} \\
        \midrule
       Method  & \texttt{Private-only} & \texttt{3LegRace} \cite{legrace_PETS_2022} & \method{} & \texttt{Private-only} & \texttt{Slalom} \cite{Slalom_ICLR_20}& \method{} \\
       Time (ms) & 1372 ($1\times$) & 237 ($6\times$) & 62 ($22\times$) & 510 ($1\times$) & 84 ($6\times$) & 20 ($25\times$) \\
        \bottomrule
    \end{tabular}
\end{table*}
\fi

\begin{table}[!htb]
    \caption{\footnotesize Running time of \method{} and other baselines (ResNet-18/CIFAR-100, $b=32$). \method{} achieves significant speedup compared to \texttt{3LegRace}, and \texttt{Slalom}. Each cell denotes (time (ms), speedup)} 
    \label{tab:runtime}
    \centering
    \small
    \begin{tabular}{c|cccc}
        \toprule
  Task & \texttt{Priv-only} & \texttt{3LegRace} & \texttt{Slalom} & \method{} \\
        \midrule
       Train & 1372 & 237 ($6\times$) & - & 62 ($22\times$) \\
       Infer. & 510 & 95 ($5\times$) & 84 ($6\times$) & 20 ($25\times$) \\
        \bottomrule
    \end{tabular}
    \vspace{-2mm}
\end{table}

\begin{table}[!htb]
    \centering
    \small
    \caption{\footnotesize Time breakdown on ResNet-18/CIFAR-100. \method{} significantly shrinks the time gap between the private and public environments.}
    \label{tab:runtimebreak}
    \begin{tabular}{c|ccc}
    \toprule
         & \MB{} & Decomp. & Parallel \MM{}/\MRS{} \\
    \midrule
        Forward & 1 ms & 3 ms & 5/16 ms\\
        Backward & 2 ms & 1 ms & 11/39 ms\\
    \bottomrule
    \end{tabular}
    \vspace{-2mm}
\end{table}

\iffalse
\begin{figure*}[!htb]
    \small
    \centering
    \includegraphics[width=.75\linewidth]{plot/runtimebreak.pdf}
    \caption{Running time breakdown of \method{} on ResNet-18/CIFAR-100. The time on \MB{} and IR decomposition is relatively small compared to \MM{} and \MRS{}. Compared to \texttt{Slalom} and \texttt{3LegRace}, \method{} significantly shrinks the time gap between the private and public environments. }
    \label{fig:runtimebreak}
    %\vspace{-2mm}
\end{figure*}
\fi

Table \ref{tab:runtime} shows the per-iteration training and inference time for \method{} and the baselines.
In both phases, owing to the effective asymmetric decomposition, the required computation in the private environment is greatly reduced.
Hence, \method{} achieves significant speedups compared to running a model solely in the private environment (\texttt{Priv-only}). 
Moreover, unlike the existing private inference method \texttt{Slalom} and training method \texttt{3Legrace}, \method{} obviates the need for frequent and costly inter-environment communication. 
Hence, \method{} offers much faster training and inference. 
Table \ref{tab:runtimebreak} further lists the running time breakdown for \method{}.
We observe that the time on \MB{} and IR decomposition is marginal compared to \MM{} and \MRS{}. 
While the theoretical computation complexity in \MM{} is much lower than \MRS{}, the real running time of \MM{} still dominates the forward and backward passes. 
Nevertheless, compared to \texttt{Slalom} and \texttt{3LegRace}, the time gap between the private and public environments is significantly shrunk.

\subsection{Protection against Attacks}\label{subsec:modelinversion}
This section evaluates \method{} against two privacy attacks: a model inversion attack called \texttt{SecretRevealer}\cite{SecretRevealer_2020_CVPR}, and a membership inference attack called \texttt{ML-Leaks} \cite{MLLeaks}. 

% \subsubsection{Model Inversion Attacks} 
\textbf{Model Inversion Attacks.} \texttt{SecretRevealer} can leverage prior knowledge, such as blurred images, when reconstructing training samples. 
% We now show \method{}'s performance against a strong model inversion attack \cite{SecretRevealer_2020_CVPR} that uses a generative model to reconstruct the training dataset with prior knowledge. 
In our case, we allow the attack to use the quantized residuals as the prior knowledge when reconstructing images.
Other model inversion attacks \cite{MIplug, MIknowledge, MI_Boundary, MI_Fidelity} need confidence scores/predicted labels from the whole model, hence they do not apply to our setting.

% Suppose $\mathcal{M}$ denotes the target model: $ \left \{ \mathcal{M}_{\text{bb}}, \mathcal{M}_{\text{main}}, \mathcal{M}_{\text{res}} \right \}$.
Following the attack protocol, we first take quantized residuals as a prior knowledge and train a generative model $G$.
Then, we optimize the latent inputs $z$ to minimize the loss on the residual model and the discriminator used to penalize unrealistic images ($L_D$): $z^{*} = \arg\min_{z} L_D(G(z)) + \lambda L_{\mathcal{M}_{\text{res}}}(G(z))$, where $\lambda$ controls the weights of $L_{\mathcal{M}_{\text{res}}}$. 
\begin{table}[!htb]
    \vspace{-3mm}
    \caption{\footnotesize \method{}'s performance against \texttt{SecretRevealer} on ResNet-18/CIFAR-100. Without DP, the attacker gains some prior information. With noise added, the quality of reconstruction degrades significantly.}
    \label{tab:attack}
    \small
    \centering
    \begin{tabular}{cc|cc|cc}
        \toprule
        \multicolumn{2}{c|}{No DP} & \multicolumn{2}{c|}{DP with $\epsilon=1.4$} & \multicolumn{2}{c}{No residuals} \\
        \midrule
         $\mathrm{SSIM}$ & $Acc_{\mathcal{M}}$ & $\mathrm{SSIM}$ & $Acc_{\mathcal{M}}$ & $\mathrm{SSIM}$ & $Acc_{\mathcal{M}}$ \\
         0.18 & 6.75\% & 0.09 & 2.13\% & 0.08 & 1\% \\
        \bottomrule
    \end{tabular}
    \vspace{-3mm}
\end{table}

Table \ref{tab:attack} shows \method{}'s performance against the model inversion attack on ResNet-18/CIFAR-100. We use the structural similarity index ($\mathrm{SSIM}$) to measure the similarity between the reconstructed and the original images \cite{wang2004image}. We also measure the target model's accuracy $Acc_{\mathcal{M}}$ given the reconstructed images as inputs. High $Acc_{\mathcal{M}}$ indicates the target model regarding the reconstructed images by the attacker close to the original training samples.
First, we observe that given quantized residuals without noise added, the attack can generate images that have slightly high accuracy on the target model. This indicates that the attacker can leverage the residual information during attacks. 
The observation is also visually reflected in reconstructed samples in Figure \ref{fig:attack}. Compared to the original training samples, some reconstructed images reveal the outline of objects (e.g., row 1, col 3).  
However, with DP, the attacker fails to use the residual information to generate images similar to the training samples. 
In particular, the accuracy of generated images on the target model is significantly reduced and also  $\mathrm{SSIM}$ between the reconstructed and the original samples.  
As a result, the attacker behaves like the case that generates images without prior knowledge of the residuals (col 3 in Table \ref{tab:attack}).

\begin{figure}[!htb]
    \centering
    \begin{subfigure}{.15\textwidth}
        \centering
        \includegraphics[width=\linewidth]{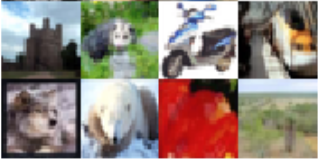}
        \caption{\footnotesize Original samples}
    \end{subfigure}
    \hfill
    \begin{subfigure}{.15\textwidth}
        \centering
        \includegraphics[width=\linewidth]{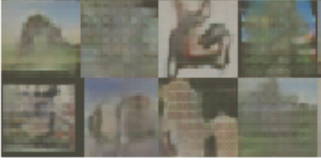}
        \caption{\footnotesize Recon. (no noise)}
    \end{subfigure}
    \hfill
    \begin{subfigure}{.15\textwidth}
        \centering
        \includegraphics[width=\linewidth]{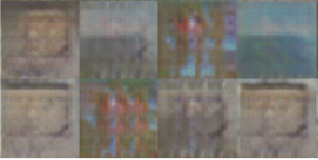}
        \caption{\footnotesize Recon. ({ $\epsilon = 1.4$})}
    \end{subfigure}
    \vspace{-1mm}
    \caption{\footnotesize Visualization of reconstructed images compared to the original training samples. With no DP, the model inversion attack recovers certain features ((b), row 1, col 3). However, when adding noise to the residual, the generated images are drastically affected. }
    \label{fig:attack}
    \vspace{-2mm}
\end{figure}

% \subsubsection{Membership Inference Attacks}
\textbf{Membership Inference Attacks.} We choose \texttt{ML-Leaks}, a strong membership inference attack that infers membership based on confidence scores \cite{MLLeaks}.
In our case, we allow the attack to use confidence scores from \MRS{}. Other membership inference attacks either required labels \cite{labelMI, LabelOnly1, LabelOnly2} or susceptible to noise \cite{MI_SP_2022}. As the predicted labels are secure with \method{} and residuals are perturbed, these methods do not apply to our setting.  
 
\noindent Following the procedure in \cite{MLLeaks}, we feed private and public samples to the shadow model, and obtain confidence score vectors from the residual model \MRS{}. We then use the vectors from public and private datasets to train an attacker model, which learns to classify whether the confidence score vector comes from public or private datasets.

\begin{table}[!htb]
    \centering
    \caption{\footnotesize Membership inference attack on Delta (ResNet-18, CIFAR-100). The DP mechanism is essential to provide further protection for the residual.}
    \label{tab:attack:mmi}
    \small
    \begin{tabular}{c|cc|cc}
    \toprule
         &  \multicolumn{2}{c|}{Attack w. $5k$ samples} & \multicolumn{2}{c}{Attack w. $10k$ samples} \\ \midrule
         & No DP & $\epsilon=1.4$ & $\epsilon=\infty$ & $\epsilon=1.4$ \\
        Acc & 0.56 & 0.52 & 0.60 & 0.55 \\
        F1 & 0.68 & 0.57 & 0.73 & 0.68 \\
    \bottomrule
    \end{tabular}
    \vspace{-3mm}
\end{table}

Table \ref{tab:attack:mmi} lists the attack performance on CIFAR-100 with ResNet-18. To align with \attMMI{}'s protocol, we use $5k$ and $10k$ samples from the training dataset as a public dataset. The rest of the samples are used as a private dataset to train the target model (using \method{}). We train the attack model for $50$ epochs, with an initial learning rate of $0.1$. \\
We observe that with perturbed residuals, attacks through \MRS{}'s outputs result in a poor performance compared to training without noise added. It indicates the DP mechanism provides further protection for the residuals, and prevents attackers from inferring membership. 
Furthermore, as the number of public samples reduces, the attack performance degrades further, implying the attack also heavily depends on prior knowledge via accessing a subset of the target dataset. 
Such an observation reveals one critical limitation of the attacks. That is, they need to get access to a subset of the target data to obtain a good estimate of the target data distributions. Otherwise, the attack performance collapses. 
However, in real scenarios, private data can be completely out of attackers' reach, rendering the target distribution's estimation impossible. As a result, the attacks can easily fail.

\section{Conclusion}\label{sec:conclusion}
We proposed a generic private training and inference framework, \method{}, with strong privacy protection, high model accuracy, and low complexity. 
\method{} decomposes the intermediate representations into asymmetric flows: information-sensitive and residual flows. We design a new low-dimensional model to learn the information-sensitive part in a private environment, while outsourcing the residual part to a large model in a public environment.
A DP mechanism and binary quantization scheme further protect residuals and improve inter-environment communication efficiency.
Our evaluations show that \method{} achieves strong privacy protection while maintaining model accuracy and computing performance.
While we evaluate \method{} in a TEE-GPU environment, \method{} can be generalized to other  setups such as federated settings, with resource-constrained client side as a private environment, and the server as a public environment. 
%In FL, we regard the resource-constrained client side as private, and the server as public. 
%Then the backbone and low-dimensional model are trained at the clients, while the residual model is offloaded to the server.
%Hence, strong privacy protection is attained while keeping low complexity in clients. 

\subsection*{Acknowledgments}

This material is based upon work supported by ONR grant N00014-23-1-2191, ARO grant W911NF-22-1-0165, Defense Advanced Research Projects Agency (DARPA) under Contract No. FASTNICS HR001120C0088 and HR001120C0160. The views, opinions, and/or findings expressed are those of the author(s) and should not be interpreted as representing the official views or policies of the Department of Defense or the U.S. Government.

{   
    \small
    \bibliographystyle{ieeenat_fullname}
    \bibliography{main}
}

% WARNING: do not forget to delete the supplementary pages from your submission 
\clearpage
\setcounter{page}{1}
\maketitlesupplementary

\section{Proofs of Theorem \ref{th:model} and Theorem \ref{th:privacy}}
In this appendix, we provide our proofs. We start with the proof of Theorem \ref{th:model} in Appendix \ref{appx:model}. Then, we provide the proof of Theorem \ref{th:privacy} in Appendix \ref{appx:privacy}. 
\subsection{Proof of Theorem \ref{th:model}}
\label{appx:model}
\iffalse
Such an adaptation comes with prior knowledge of the ranks of inputs and outputs. For an original convolution layer with kernel $\tW\in\sR^{n\times c\times k\times k}$, given input $\tX\in\sR^{c\times h\times w}$, output is calculated as $\tY = \tW \circledast \tX$.
We can also write $\tY$ in a matrix form as\footnote{See Appendix \ref{} for further details} $Y = W \cdot X$, where $W\in\sR^{n\times ck^2}$, and $X\in\sR^{ck^2\times hw}$.
On the other hand, we can calculate the output for the new convolution layer as $\tY^{'} = \tW^1 \circledast \tW^2 \circledast \tX$.
Similarly, a matrix form can be written as $Y^{'} = W^2 \times W^1 \times X$, where $W^2\in\sR^{n\times q}$ and $W^1\in\sR^{q\times ck^2}$.
Given $Rank(Y) = q$, there exists an solution $W^1, W^2$ such that \norm{Y - Y^{'}} = 0, namely $Y^{'}$ can be the exact low-rank presentation of $Y$.
Therefore, with low-rank inputs, we can design a low-rank model as in Figure \ref{fig:model} to preserve information as in the original layer.
\fi
We now prove Theorem \ref{th:model} showing the feasibility of designing a low-dimensional layer.

\begin{appxtheorem}
    For a convolution layer with weight $\tW\in\sR^{n\times c\times k\times k}$ with an input $\tX$ with rank $r$ and output $\tY$ with rank $q$, there exists an optimal $\tW^{(1)}\in\sR^{q\times c\times k\times k}, \tW^{(2)}\in\sR^{n\times q\times 1\times 1}$ in the low-dimensional layer such that the output of this layer denoted as $\tY^{'}$ satisfies
\begin{align}
    \left \| \tY - \tY^{'}  \right \| = 0.
\end{align}
\end{appxtheorem}

\begin{proof}
Given an input tensor $\tX\in\sR^{c\times h\times w}$ with rank $r$, $\tW\in\sR^{n\times c\times k\times k}$, and output $\tY\in\sR^{n\times h\times w}$ with rank $q$, we can write the convolution $\tY = \tW \circledast \tX$ in a matrix form as 
\begin{align}
Y = W \cdot X,
\end{align}
where $W\in\sR^{n\times ck^2}$, $X\in\sR^{ck^2\times hw}$ \cite{img2col_2006,img2col_2017,legrace_PETS_2022}. Similarly for the low-dimensional convolution layer, we can write its output $\tY^{'}$ as follows 
\begin{align*}
Y^{'} = W^{(2)} \cdot W^{(1)} \cdot X,
\end{align*}
where $W^{(2)}\in\sR^{n\times q}$ and $W^{(1)}\in\sR^{q\times ck^2}$.

According to Theorem 2 in \cite{legrace_PETS_2022}, the output's rank is decided by the rank of $X$. Therefore, we can decompose $X$ using SVD as 
\begin{equation*}
    X = U \cdot V,
\end{equation*}
where  $U\in\sR^{ck^2 \times q}$ is a matrix with orthogonal columns, and $V\in\sR^{q \times hw}$ is a matrix with orthogonal rows. With the low-rank decomposition, we can express the difference between the outputs $Y$ and $Y^{'}$ as 
\begin{equation*}
\begin{split}
    Y - Y^{'} &= W \cdot U \cdot V - W^{(2)} \cdot W^{(1)} \cdot U \cdot V \\
    &= (W\cdot U - W^{(2)}\cdot W^{(1)} \cdot U) \cdot V.
\end{split}
\end{equation*}
Let $Z = W\cdot U - W^{(2)} \cdot W^{(1)} \cdot U$, owing to the low-rank of the matrix $U$, the original kernel matrix $W$ is reduced to a low-dimensional form. Therefore, we can express the kernel matrix $W$ as follows 
\begin{align}
W = W^U + R,
\end{align}
where $W^U = W \times UU^{*}$ is a low-rank matrix obtained based on the principal components of $U$, while $R$ is a residual matrix based on the principal components that are orthogonal to $U$. That is, $R \cdot U=0$. \\ Then the difference $Z = W\cdot U -W^{(2)}\cdot W^{(1)}\cdot U$ can be written as follows 
\begin{equation*}
\begin{split}
    Z &= \left [ (W^U + R) - W^{(2)} \cdot W^{(1)} \right ] \cdot U \\
      &= (W^U - W^{(2)}\cdot W^{(1)}) \cdot U.
\end{split}
\end{equation*}

Since $\mathtt{Rank}~(W^U) \leq q$, then there is a solution such that
\begin{equation*}
    \min_{W^1, W^2} \left \| W^U - W^{(2)}\cdot W^{(1)} \right \| = 0,
\end{equation*}
where the pair $W^{(1)}, W^{(2)}$ is one of the rank-$q$ decompositions of $W^U$.
Hence, $\min_{W^1, W^2} \left \| \tY - \tY^{'} \right \| = 0$.
\end{proof}

\subsection{Proof of Theorem \ref{th:privacy}}
\label{appx:privacy}
Next, we provide the proof of Theorem \ref{th:privacy}.
\begin{appxtheorem}
\method{} ensures that the perturbed residuals 
%in Equation (\ref{eq:binquant})
and operations in the public environment satisfy $(\epsilon,\delta)$-DP given noise $\tN \sim \mathcal N (0, 2C^2\cdot\log{(2/\delta') / \epsilon'} \cdot \tI)$ given sampling probability $p$, and $\epsilon=\log{(1 + p(e^{\epsilon'}-1))}, \delta=p\delta'$.
\end{appxtheorem}

\begin{proof}
The proof relies on Theorem 6 in \cite{AproximatedDP}, the subsampling Theorem in \cite{PrivacyAmp} (Theorem 9) (replicated in Theorem 3 and 4), and the post-processing rule of DP.

First, we show that the mechanism is $(\epsilon, \delta)$-DP if $\sigma^2 = 2C^2\cdot\log{(2/\delta') / \epsilon'}$. 
% With subsampling probability $p$, we know the probability of $\tX_{\text{res}}$ and $\tX'_{\text{res}}$ are different is $p$. That is, $P[\tX_{\text{res}} \neq \tX'_{\text{res}}] = p$.
%Considering the case that $\tX_{\text{res}}$ and $\tX'_{\text{res}}$ are different, we use a Gaussian mechanism to perturb the residuals. 
Since the mechanism of obtaining \IRN{} is a Gaussian mechanism, and the variance is $2C^2\cdot\log{(2/\delta') / \epsilon'}$, then it follows that the mechanism ensures $(\epsilon,\delta)$-DP by Theorem 3 and Theorem 4. 
%\IRN{} is $(\epsilon, \delta)$ differentially private given sensitivity $\Delta \leq C$.

The residual model \MRS{} in the public environment and outputs after \MRS{} are the post-processing of \IRN{}. Since the post-processing does not affect the DP budget \cite{DPFoundation}, any operation in the public environment has the same privacy budget as the Gaussian mechanism.  
\end{proof}

\begin{appxtheorem}
    (Theorem 6 in \cite{AproximatedDP}) Given noise with $\sigma^2=2C^2\cdot \log{(2/\delta)}/\epsilon$, the Gaussian mechanism is $(\epsilon, \delta)$-DP. 
\end{appxtheorem}

\begin{appxtheorem}
    (Theorem 9 \cite{PrivacyAmp}) Given a randomized mechanism $\mathcal{M}'$ with privacy parameter $(\epsilon', \delta')$, and $\mathcal{M}$ with sampling probability $p$, for any $\epsilon' > 0, \delta' > 0$, we have $\epsilon = \log{(1 + p(e^{\epsilon'}-1))}$, and $\delta = p\delta'$. 
\end{appxtheorem}

\section{Ablation Study}\label{appx:ablation}
In this section, we conduct three important ablation studies. The first one explores merging logits with scaling factor, whereas the second investigates more perturbation effects on the overall model performance and finally the third one investigates the effects of the binary quantization. 

\subsection{Logits Merging with Scaling}
\label{appx:ablation:logits}

In the main paper, we directly add the logits from \MM{} and \MRS{} to obtain the final prediction (See Figure \ref{fig:overview:delta}). 
In this appendix, we explore a different way of merging logits. 
Specifically, given logits vector from \MM{} and \MRS{}: $\bm{z}_{\text{main}}$, $\bm{z}_{\text{res}}$, we add a scaling coefficient $\alpha$ during merging as

\vspace{-3mm}

\begin{equation*}
    \bm{z}_{\text{tot}} = \bm{z}_{\text{main}} + \alpha \cdot \bm{z}_{\text{res}}.
\end{equation*}

The scaling factor controls the weight of \MRS{}'s prediction. Since \MRS{} only contains residual information, its prediction might conflict with \MM{} when residuals contain little information. 
With the scaling factor, potential conflicts between \MM{} and \MRS{} can be mitigated. 

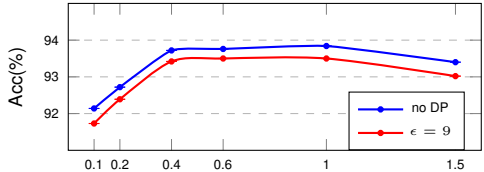
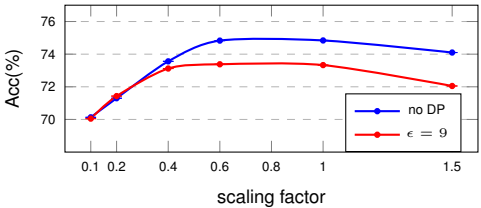
\begin{figure}[!htb]
    \centering
    \begin{subfigure}{.45\textwidth}
    \begin{tikzpicture}
        \begin{axis}[
            width=.9\linewidth, height=.45\linewidth,
            xmin=0, xmax=1.6,
            xtick={ 0.1, 0.2, 0.4, 0.6, 1, 1.5 },
            xticklabels={ 0.1, 0.2, 0.4, 0.6, 1, 1.5 },
            ymin=91, ymax=95,
            ylabel={\scriptsize Acc(\%)},
            ytick={92, 93, 94},
            yticklabels={92, 93, 94},
            ymajorgrids,
            major grid style={line width=.2pt, draw=black!30, dashed},
            % extra y ticks={92, 94},
            % extra y tick labels={92, 94},
            extra tick style={major grid style=red, dashed},
            ticklabel style={font=\tiny},
            legend style={at={(0.96,0.4)}},
        ]
    
            % No DP
            \addplot+[ 
                thick, blue, smooth, tension=0.3, mark=*, mark options={solid, scale=0.4},
                error bars/.cd, y fixed, y dir=both, y explicit
            ] 
            table [x=x, y=y, y error plus=error1, y error minus=error2, col sep=comma]{
                x,    y,      error1,    error2
                0.1,    92.14,   0.0,      0.0
                0.2,    92.72,   0.0,      0.0
                0.4,    93.72,   0.0,      0.0
                0.6,    93.76,   0.0,      0.0
                1.0,    93.84,   0.0,      0.0
                1.5,    93.4,   0.0,      0.0
            };
            \addlegendentry{\tiny no DP}    

            % ep=2
            \addplot+[ 
                thick, red, smooth, tension=0.5, mark=*, mark options={solid, scale=0.4},
                error bars/.cd, y fixed, y dir=both, y explicit
            ] 
            table [x=x, y=y, y error plus=error1, y error minus=error2, col sep=comma]{
                x,    y,      error1,    error2
                0.1,    91.73,   0.0,      0.0
                0.2,    92.39,   0.0,      0.0
                0.4,    93.42,   0.0,      0.0
                0.6,    93.5,   0.0,      0.0
                1.0,    93.50,   0.0,      0.0
                1.5,    93.02,   0.0,      0.0
            };
            \addlegendentry{\tiny $\epsilon=9$}  
        \end{axis}
    \end{tikzpicture}
    \caption{ResNet-18/CIFAR-10}
    \label{fig:ab:logits:cifar10}
    \end{subfigure}

    \vspace{2mm}
    
    \begin{subfigure}{.45\textwidth}
    \begin{tikzpicture}
        \begin{axis}[
            width=.9\linewidth, height=.45\linewidth,
            xmin=0, xmax=1.6,
            xlabel={\scriptsize scaling factor},
            xtick={ 0.1, 0.2, 0.4, 0.6, 0.8, 1, 1.5 },
            xticklabels={ 0.1, 0.2, 0.4, 0.6, 0.8, 1, 1.5 },
            ymin=68, ymax=77,
            ylabel={\scriptsize Acc(\%)},
            ytick={70, 72, 74, 76},
            yticklabels={70, 72, 74, 76},
            ymajorgrids,
            major grid style={line width=.2pt, draw=black!30, dashed},
            % extra y ticks={92, 94},
            % extra y tick labels={92, 94},
            extra tick style={major grid style=red, dashed},
            ticklabel style={font=\tiny},
            legend style={at={(0.96,0.4)}},
        ]
    
            % No DP
            \addplot+[ 
                thick, blue, smooth, mark=*, mark options={solid, scale=0.4},
                error bars/.cd, y fixed, y dir=both, y explicit
            ] 
            table [x=x, y=y, y error plus=error1, y error minus=error2, col sep=comma]{
                x,    y,      error1,    error2
                0.1,    70.13,   0.0,      0.0
                0.2,    71.30,   0.0,      0.0
                0.4,    73.56,   0.0,      0.0
                0.6,    74.83,   0.0,      0.0
                1.0,    74.84,   0.0,      0.0
                1.5,    74.1,   0.0,      0.0
            };
            \addlegendentry{\tiny no DP}    

            % ep=2
            \addplot+[ 
                thick, red, smooth, mark=*, mark options={solid, scale=0.4},
                error bars/.cd, y fixed, y dir=both, y explicit
            ] 
            table [x=x, y=y, y error plus=error1, y error minus=error2, col sep=comma]{
                x,    y,      error1,    error2
                0.1,    70.05,   0.0,      0.0
                0.2,    71.43,   0.0,      0.0
                0.4,    73.12,   0.0,      0.0
                0.6,    73.38,   0.0,      0.0
                1.0,    73.33,   0.0,      0.0
                1.5,    72.05,   0.0,      0.0
            };
            \addlegendentry{\tiny $\epsilon=9$}  
        \end{axis}
    \end{tikzpicture}
    \caption{ResNet-18/CIFAR-100}
    \label{fig:ab:logits:cifar100}
    \end{subfigure}
    \vspace{-2mm}
    \caption{\footnotesize Ablation study on merging logits with scaling. Logits merging with a small $\alpha$ limits the useful information from the residual model, incurring accuracy drops. Merging with a large $\alpha (> 1)$ also incurs accuracy drop as \MRS{} can overshadow \MM{} in the final prediction. 
    % There is a large adjustment space for $\alpha$, which leads to optimal performance.
    }
    \label{fig:ab:logits}
    \vspace{-4mm}
\end{figure}

Figure \ref{fig:ab:logits} shows the accuracy versus $\alpha$ for ResNet-18 on CIFAR-10/100. We next make the following observations.

\begin{itemize}
    \item With small $\alpha$, there is a noticeable accuracy drop. The reason is that small $\alpha$ reduces the weight of \MRS{}'s logits, limiting information from the residual path. As a result, \MRS{} barely provides performance improvements.

    \item As $\alpha$ increases, \MRS{}'s prediction weighs more, thereby boosting the performance of the overall model. We can also observe that there is a large adjustment space for $\alpha$, which leads to optimal performance. 

    \item With further large $\alpha$ ($\alpha >1$), \MRS{} becomes more and more dominant and dominates the main model's prediction. The overall performance again decreases. 
\end{itemize}

Therefore, the weight of \MRS{}'s logits affects the overall model performance. \method{} in the main paper assigns equal weights for \MM{} and \MRS{} ($\alpha=1$), which strikes an optimal balance between predictions from those two models. 

\subsection{More Effects of Perturbation}
\label{appx:ablation:residual}

In this ablation study, we investigate the potential adverse effects of the perturbed residuals on overall performance. As we add very large noise on residuals, information in residuals is significantly perturbed. As a result, the residual model \MRS{} can cause conflicts with the main model, rather than provide additional beneficial information for the final prediction. 

Figure \ref{fig:ab:residual} shows the accuracy of ResNet-18 on CIFAR-10/100 with small privacy budgets. With small $\epsilon$, the noise for perturbation is very large, thereby making the residual model \MRS{} unable to extract useful information from residuals. And it further affects the overall model accuracy. In particular, the final model accuracy can be even lower than without \MRS{} (red dashed line in Figure \ref{fig:ab:residual}).
Note that this ablation study mainly aims to investigate more effects of perturbation under very small $\epsilon$, but such strong privacy constraints are usually not considered in real scenarios. 
Furthermore, even in this case, with very strict privacy constraints, users can train \MB{} and \MM{} only. Owing to the effective asymmetric decomposition, \MB{} and \MM{} still give reasonable accuracy without incurring prohibitive costs in the private environment.  

\begin{figure}[!htb]
    \centering
    \begin{subfigure}{.45\textwidth}
    \centering
    \begin{tikzpicture}
        \begin{axis}[
            width=.9\linewidth, height=.4\linewidth,
            xmin=0, xmax=1.5,
            xtick={ 0.05, 0.3, 1, 1.4 },
            xticklabels={ 0.05, 0.3, 1, 1.4 },
            ymin=90, ymax=93,
            ylabel={\scriptsize Acc(\%)},
            ytick={90, 91, 92},
            yticklabels={90, 91, 92},
            ymajorgrids,
            major grid style={line width=.2pt, draw=black!30, dashed},
            extra y ticks={92},
            extra y tick labels={92},
            extra tick style={major grid style=red, dashed},
            ticklabel style={font=\tiny},
        ]
    
            % No DP
            \addplot+[ 
                thick, red, smooth, mark=*, mark options={solid, scale=0.4},
                error bars/.cd, y fixed, y dir=both, y explicit
            ] 
            table [x=x, y=y, y error plus=error1, y error minus=error2, col sep=comma]{
                x,    y,      error1,    error2
                % (91.02)
                % 0.1,    91.02,   0.0,      0.0
                % (91.06, 91.82)
                0.05,    91.44,   0.0,      0.0
                0.3,    92.0,    0.0,      0.0
                1,    92.2,    0.0,      0.0
                1.4,      92.4,    0.0,      0.0
            };

            \node[fill=white, draw=black!30, rounded corners] (MB) at (axis cs:1.2,91.3){\tiny {$\mathcal{M}_{\text{bb,main}}$ only}};
            \draw[->] (MB.north west) -- +(-1mm, 1mm);
        \end{axis}
    \end{tikzpicture}
    \caption{ResNet-18/CIFAR-10}
    \label{fig:ab:residual:cifar10}
    \end{subfigure}

    \vspace{2mm}
    
    \begin{subfigure}{.45\textwidth}
    \centering
    \hspace{-2mm}
    \begin{tikzpicture}
        \begin{axis}[
            width=.9\linewidth, height=.4\linewidth,
            xmin=0, xmax=1.5,
            xlabel={\scriptsize $\epsilon$},
            xtick={ 0.05, 0.3, 1, 1.4 },
            xticklabels={ 0.05, 0.3, 1, 1.4 },
            ymin=66, ymax=73,
            ylabel={\scriptsize Acc(\%)},
            ytick={67, 71},
            yticklabels={67, 71},
            ymajorgrids,
            major grid style={line width=.2pt, draw=black!30, dashed},
            extra y ticks={69.8},
            extra y tick labels={69.8},
            extra tick style={major grid style=red, dashed},
            ticklabel style={font=\tiny},
        ]
    
            % No DP
            \addplot+[ 
                thick, red, smooth, mark=*, mark options={solid, scale=0.4},
                error bars/.cd, y fixed, y dir=both, y explicit
            ] 
            table [x=x, y=y, y error plus=error1, y error minus=error2, col sep=comma]{
                x,    y,      error1,    error2
                % (67.82 67.24)
                % 0.1,    67.51,   0.0,      0.0
                % (68.06, 68.31)
                0.05,    68.19,   0.0,      0.0
                0.3,    69.05,   0.0,      0.0
                % (70.15, 69.91)
                1,    70.03,   0.0,      0.0
                1.4,    71.4,    0.0,      0.0
            };

            \node[fill=white, draw=black!30, rounded corners] (MB) at (axis cs:1.2,68){\tiny {$\mathcal{M}_{\text{bb,main}}$ only}};
            \draw[->] (MB.north west) -- +(-1mm, 1mm);
        \end{axis}
    \end{tikzpicture}
    \caption{ResNet-18/CIFAR-100}
    \label{fig:ab:residual:cifar100}
    \end{subfigure}
    \vspace{-3mm}
    \caption{\footnotesize Effects of large perturbation on overall performance. \MRS{} under strict privacy constraints can result in an overall accuracy even lower than without \MRS{}. In this case, users can train \MB{} and \MM{} only, which achieves reasonable accuracy owing to the asymmetric decomposition.}
    \vspace{-4mm}
    \label{fig:ab:residual}
\end{figure}
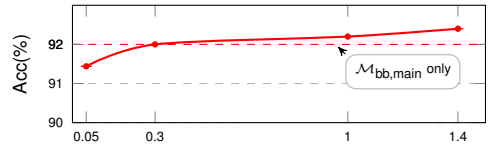
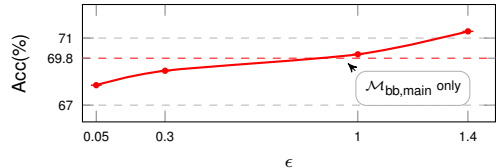

\subsection{Effects of Binary Quantization}
\label{appx:ablation:quantization}

This ablation study analyzes how the binary quantization affects the final model's performance. As elaborated in Section \ref{sec:method:quant}, the binary quantization reduces the communication cost when sending residuals to the public environment. \\
We train ResNet-18 on CIFAR-10 and CIFAR-100 with and without binary quantization, and report their results in Table \ref{tab:ablation:quant}. The results show that binary quantization does not significantly affect the accuracy of the final model under different privacy budgets.

\begin{table*}[!htb]
    \centering
    \begin{tabular}{c|cc|cc}
    \toprule
        & { CIFAR10:} $\epsilon=1.4$ & $\epsilon=\infty$ & { CIFAR100:} $\epsilon=1.4$ & $\epsilon=\infty$ \\ \midrule
        With Quantization & $92.4 \pm 0.3$ & $93.7 \pm 0.3$ & $71.4 \pm 0.1$ & $74.8 \pm 0.3$ \\
        No Quantization & $92.8 \pm 0.2$ & $94 \pm 0.2$ & $72 \pm 0.4$ & $75.1 \pm 0.2$ \\
    \bottomrule
    \end{tabular}
    \vspace{-3mm}
    \caption{\footnotesize Ablation study w/ and w/o quantization with ResNet-18. This demonstartes that the quantization does not significantly reduce model performance.}
    \label{tab:ablation:quant}
    \vspace{-3mm}
\end{table*}

\section{Experiment for Asymmetric Structure}
In this appendix, we provide the experimental details for the asymmetric structure of the IRs in Section \ref{sec:motivation}. \\
We train ResNet-18 with ImageNet  and the hyperparameters listed in Table \ref{tab:hparam:motivation}. When the training is complete, we analyze the asymmetric structures in the validation dataset. \\
Specifically, we extract the intermediate features after the first convolutional layer in ResNet-18, then use SVD and DCT to analyze the channel and spatial correlation as in Section \ref{sec:motivation}.
For DCT, since the feature size after the convolutional layer is $56\times 56$, we use $14\times 14$ block-wise DCT. \\
We compute the relative error $\frac{\left \| \tX - \tX_{\textrm{lr}} \right \|}{ \left \| \tX \right \| }$ and $\frac{\left \| \tX - \tX_{\text{lf}} \right \|}{ \left \| \tX \right \| }$ for each input, and average the ratio across the entire validation datasets.
By varying the number of principal channels in $\tX_{\text{lr}}$, $r$, and the number of low-frequency components in $\etX_{\text{lf}}$, $t'^2$, we obtain the results in Figure \ref{fig:motivation}.

\begin{table}[!htb]
    \centering
    \caption{\footnotesize Hyperparameters in investigating asymmetric structures in CNNs}
    \label{tab:hparam:motivation}
    \small
    \begin{tabular}{c|c|c|c|c|c}
    \toprule
       batch size & epochs & lr & wd & momem. & lr scheduler \\
    \midrule
       128 & 100 & 0.1 & 1e-4 & 0.9 & cosine anneal  \\
    \bottomrule
    \end{tabular}
    \vspace{-3mm}
\end{table}

\subsection{Asymmetric Structure in Language Models}
\label{appx:asymmetricNLP}
The asymmetric structure of the intermediate representations is not only observed in computer vision models but also in language models.
In this appendix, we show the asymmetric structure of word embedding vectors in language models. 

We use Word2Vec \cite{Word2Vec_2013_arXiv} and the word embedding layer in BERT \cite{BERT_2018_arXiv} to generate embedding vectors. 
Figure \ref{fig:lm:origtext} shows an example text with 80 words. We feed the text to Word2Vec and BERT word embedding layer, obtaining embedding vectors. 
We group the embedding vectors as a matrix with each one stored in one row.
With the embedding matrix obtained from Word2Vec and BERT, we respectively apply SVD to the matrix and compute their singular values, as shown in Figure \ref{fig:lm:singularvalue}.
We can easily observe that the decay of the singular values follows an exponential manner for both Word2Vec and the BERT embedding layer, indicating high correlations among the embedding vectors. 
With the principal vector after SVD, we further use the first 16 principal vectors ($1/5$ of the total vectors) from Word2Vec embedding and reconstruct an approximated matrix, where each row approximates the original embedding vector.
Then, we reconstruct the text using Vec2Word, as shown in Figure \ref{fig:lm:origtext}.  
We observe that the approximated text is almost the same as the original one even with only $1/5$ principal vectors (difference highlighted in bold red). 

\begin{figure}[!htb]
\centering
\begin{subfigure}{.48\textwidth}
\centering
\begin{tikzpicture}
    \begin{axis}[
        xmin = -1, xmax = 80,
        ymin = 0, ymax = 11,
        xlabel = {\scriptsize Singular value indices},
        ylabel = {\scriptsize Singular values}, 
        xtick = { 0, 20, 40, 60, 80 },
        xticklabels = { \tiny 0, \tiny 20, \tiny 40, \tiny 60, \tiny 80 },
        ytick = { 0, 2, 4, 6, 8 },
        yticklabels = { \tiny 0, \tiny 2, \tiny 4, \tiny 6, \tiny 8 },
        width=0.95\textwidth, height=.4\textwidth,
        line width=0.1pt,
        ]

        \addplot[ 
        thick, blue,
        mark=*, mark options={solid, scale=0.2}] file [skip first] { ./data/s_bert.dat };

        \addplot[ 
        thick, red,
        mark=*, mark options={solid, scale=0.2}] file [skip first] { ./data/s_word2vec.dat };

        \legend{
            \tiny BERT, \tiny Word2Vec
        }
    \end{axis}
\end{tikzpicture}
\caption{\footnotesize Singular values in word embedding vectors from BERT and Word2Vec.}
\label{fig:lm:singularvalue}
\end{subfigure}

\vspace{2mm}

\begin{subfigure}{.49\textwidth}
\centering
\begin{tikzpicture}
    \node[rounded corners, draw=green, fill=green!10, text width=.95\textwidth]{\baselineskip=6pt \scriptsize Large Language Models are foundational machine learning models that use deep learning algorithms to process and understand natural language. These models are trained on massive amounts of text data to learn patterns and entity relationships in the language. Large Language Models can perform many types of language tasks, such as translating languages, analyzing sentiments, chatbot conversations, and more. They can understand complex textual data, identify entities and relationships between them, and generate new text that \textcolor{red}{\textbf{is}} coherent and grammatically accurate. \par};
\end{tikzpicture}
\caption{Original text}
\label{fig:lm:origtext}
\end{subfigure}

\begin{subfigure}{.49\textwidth}
\centering
\begin{tikzpicture}
    \node[rounded corners, draw=green, fill=green!10, text width=.95\textwidth]{\baselineskip=6pt \scriptsize Large Language Models are foundational machine learning models that use deep learning algorithms to process and understand natural language. These models are trained on massive amounts of text data to learn patterns and entity relationships in the language. Large Language Models can perform many types of language tasks, such as translating languages, analyzing sentiments, chatbot conversations, and more. They can understand complex textual data, identify entities and relationships between them, and generate new text that \textcolor{red}{\textbf{are}} coherent and grammatically accurate. \par};
\end{tikzpicture}
\caption{approximated text with $1/5$ principal vectors from Word2Vec.}
\label{fig:lm:approxtext}
\end{subfigure}
\caption{\footnotesize Asymmetric structure in language models. Embeddings in language models also have a highly asymmetric structure. An approximated text with only $1/5$ principal vectors is almost the same as the original.}
\vspace{-4mm}
\end{figure}

\section{Model and Training Details}
In this appendix, we provide the model architectures and the hyperparameters of the experiments presented in Section \ref{sec:eval}. 

\begin{table}[!htb]
    \centering
    \small
    \caption{\footnotesize Model parameters of \MM{} for ResNet-18 and ResNet-34 with 8 principal channels in \IRM{}.}
    \label{tab:model:lowdim}
    \begin{tabular}{cccc|cccc}
    \toprule
        \multicolumn{4}{c|}{ResNet-18} & \multicolumn{4}{c}{ResNet-34} \\ 
       \emph{Resblock}  & $n$ & $k$ & $q$ & \emph{Resblock}  & $n$ & $k$ & $q$ \\ \midrule
        1,2 & 64  & 3 & 16   & 1-3 & 64  & 3 & 16  \\
        3,4 & 256 & 3 & 32   & 4-9 & 256 & 3 & 32 \\
        5,6 & 512 & 3 & 64   & 10-11 & 512 & 3 & 64 \\
    \bottomrule
    \end{tabular}
    \vspace{-4mm}
\end{table}

\subsection{Model Architectures}
\label{appx:arch}
Since the backbone model \MB{} and high-dimensional model \MRS{} combined is just the original model, in this section, we omit their architecture details and only provide the architecture of \MM{}.

\vspace{-3mm}

\begin{figure}[!htb]
    \centering
    \begin{subfigure}{0.45\linewidth}
        \centering
        \begin{tikzpicture}[node distance=1cm, auto,]
            \node[above=10mm] (input) {\scriptsize Input (rank: $r$)};
            \node[conv, fill=green!10] (conv1) {\scriptsize Conv $k\times k$};
            \node[conv, fill=green!10, below=5mm of conv1] (conv2) {\scriptsize Conv $k\times k$};
            \node[port, below=5mm of conv2] (add) {\scriptsize +};
            \node[below=5mm of add] (output) {\scriptsize Output (rank: $q$)};
    
            \node[text width=12mm, left=0mm of conv1] {\tiny $n$ kernels};
            \node[text width=12mm, left=0mm of conv2] {\tiny $n$ kernels};
    
            \draw[->]    (input) -- (conv1);
            \draw[->]    (conv1) -- (conv2);
            \draw[->]    (conv2) -- (add);
            \draw[->]    (add) -- (output);
            % \draw[->]    (input) to [out=360,in=360] (add);
            \draw[->, rounded corners]    (input.south) |- +(10mm, -2mm) |- (add.east);
        \end{tikzpicture}
        \caption{Original \emph{Resblock}($n, k$)}
        \label{fig:model:arch1}
    \end{subfigure} \quad
    \begin{subfigure}{0.45\linewidth}
        \centering
        \begin{tikzpicture}[node distance=1cm, auto,]
            \node[above=6mm] (input) {\scriptsize Input (rank: $r$)};
            \node[conv, fill=red!10] (conv1) {\scriptsize Conv $k\times k$};
            \node[conv, fill=red!10, below=1mm of conv1] (conv1_1) {\scriptsize Conv $1\times 1$};
            \node[conv, fill=red!10, below=3mm of conv1_1] (conv2) {\scriptsize Conv $k\times k$};
            \node[conv, fill=red!10, below=1mm of conv2] (conv2_1) {\scriptsize Conv $1\times 1$};
            \node[port, below=3mm of conv2_1] (add) {\scriptsize +};
            \node[below=3mm of add] (output) {\scriptsize Output (rank: $q$)};
    
            \node[text width=12mm, left=0mm of conv1] {\tiny $q$ kernels};
            \node[text width=12mm, left=0mm of conv1_1] {\tiny $n$ kernels};
            \node[text width=12mm, left=0mm of conv2] {\tiny $q$ kernels};
            \node[text width=12mm, left=0mm of conv2_1] {\tiny $n$ kernels};
    
            \draw[->]    (input) -- (conv1);
            \draw[->]    (conv1) -- (conv1_1);
            \draw[->]    (conv1_1) -- (conv2);
            \draw[->]    (conv2) -- (conv2_1);
            \draw[->]    (conv2_1) -- (add);
            \draw[->]    (add) -- (output);
            % \draw[->]    (input) to [out=360,in=360] (add);
            \draw[->, rounded corners]    (input.south) -- ++(0mm, -1mm) -- +(10mm, 0mm) |- (add.east);
        \end{tikzpicture}
        \caption{Low-dim \emph{Resblock}($n, k, q$)}
        \label{fig:model:arch2}
    \end{subfigure}
    \caption{\footnotesize The original \emph{Resblock} and low-dimensional \emph{Resblock}. Non-linear activation functions and batchnorm are not shown for simplicity.}
    \label{fig:model:arch}
    \vspace{-4mm}
\end{figure}
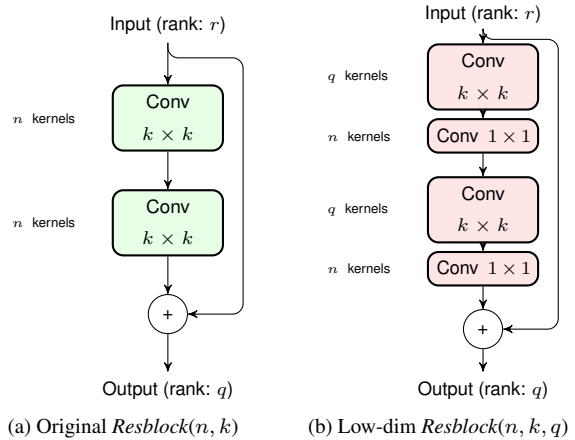
The original and low-dimensional \emph{Resblocks} are shown in Figure \ref{fig:model:arch}. The design of the low-dimensional \emph{Resblock}  follows Figure \ref{fig:model} and Theorem \ref{th:model}.

Table \ref{tab:model:lowdim} lists the details of model parameters of \MM{} for ResNet-18 and ResNet-34.
Given a \emph{Resblock} with $3\times 3$ kernels and input rank $r$, an output with rank $2r$ is sufficient to preserve most information in the principal channels \cite{legrace_PETS_2022}.
Hence, for ResNet-18 and ResNet-34, we let $q=2r$.

\iffalse
\subsection{Model Complexities}
\label{appx:complexity}
In this appendix, we present the complexity breakdown of \method{}.

\begin{table}[!htb]
    \centering
    \small
    \caption{\footnotesize MACs of ResNet-18 (CIFAR) and ResNet-34 (ImageNet) during inference with batch size 1.}
    \label{tab:model:complexity}
    \begin{tabular}{cccc|cccc}
    \toprule
         \multicolumn{4}{c|}{ResNet-18 (CIFAR)} & \multicolumn{4}{c}{ResNet-34 (ImageNet)} \\ \midrule
        \MB{}+\MM{} & SVD & DCT & \MR{} & \MB{}+\MM{} & SVD & DCT & \MR{} \\ \midrule
        48.3 M & 0.52 M & 0.26 M & 547M & 437 M & 1.6 M & 0.7 M & 3.5G \\
    \bottomrule
    \end{tabular}
\end{table}

Table \ref{tab:model:complexity} lists the theoretical computation complexities (by MACs) of ResNet-18 on CIFAR, and ResNet-34 on ImageNet during the inference phase with a batch of size $1$.
For SVD, we use an approximation algorithm \cite{legrace_PETS_2022} that only computes the first $r$ principal channels. \\
First, we can see that SVD and DCT only account for a very small fraction of the total computations, which aligns with the real running time in Figure\ref{fig:runtimebreak}.
On the other hand, compared to MACs in \MR{}, the computation complexity of \MB{}+\MM{} is much smaller, only accounting for $\sim 10\%$ of MACs in \MR{}.
This shows that the asymmetric decomposition and low-dimensional model effectively reduce the computation cost of the resource-constrained private environments.
\fi

\subsection{Hyperparameters in the Main Experiments}\label{appx:exp:hparam}
For the privacy parameters, we set $\delta$ as \num{e-6} for all datasets. \\
Table \ref{tab:hparam:cifar} and \ref{tab:hparam:imagenet} list hyperparameters in training ResNet-18 on CIFAR-10/100, and ResNet-18/34 on ImageNet\footnote{The DCT block size is chosen by trading off DCT computation complexity and the effectiveness of the low-frequency approximation. DCT with too small blocks does not effectively extract the low-frequency components, whereas larger block sizes are computationally-intensive.}.
\begin{table}[!htb]
    \centering
    \small
    \caption{\footnotesize Hyperparameters in training ResNet-18 on CIFAR-10/100.}
    \label{tab:hparam:cifar}
    \begin{tabular}{ccccc|cc}
    \toprule
       epochs  & $b$ & $\mathrm{lr}$ & $\mathrm{wd}$ & \texttt{orth reg} & $r$ & $t/t'$ \\ \midrule
        150    & 64  & 0.1  & 2e-4 & 8e-4     & 8   & 16/8  \\
    \bottomrule
    \end{tabular} \\
    \footnotesize{ $b$: batch size, $\mathrm{lr}$: initial learning rate, $\mathrm{wd}$: weight decay. } \\
    \footnotesize{ $r$: \#principal channels in \IRM{}, $t/t'$: DCT/IDCT block sizes.} \\
    \footnotesize{ \texttt{orth reg}: kernel orthogonalization regularization.}
    \vspace{-5mm}
\end{table}

\begin{table}[!htb]
    \centering
    \small
    \caption{\footnotesize Hyperparameters in training ResNet-18/34 on ImageNet.}
    \label{tab:hparam:imagenet}
    \begin{tabular}{ccccc|cc}
    \toprule
       epochs  & $b$ & $\mathrm{lr}$ & $\mathrm{wd}$ & \texttt{orth reg} & $r$ & $t/t'$ \\ \midrule
        100    & 256 & 0.1  & 2e-5 & 0 & 12   & 14/7 \\
    \bottomrule
    \end{tabular}
    \vspace{-4mm}
\end{table}

\iffalse
For noise on logits $\sigma_{\texttt{logits}}$, we compute it based on the standard Laplacian DP mechanism. Specifically, given batch sampling probability $p=b/N$, where $N$ denotes the total number of samples, $\epsilon_{\texttt{logits}} = p\cdot\frac{\Delta_{\texttt{logits}}}{\sigma_{\texttt{logits}}}$. 
Sine the logits are $\ell_1$ normalized after \texttt{Softmax}, we have $\Delta_{\texttt{logits}} = 2$. This follows since for any two output vectors of the \texttt{Softmax} layer denoted as $\vo_1$ and $\vo_2$, we have
\begin{align}
\left \| \vo_1 - \vo_2 \right \|_1 \leq \left \| \vo_1 \right \|_1 + \left \| \vo_2 \right \|_1 = 2,
\end{align}
where the inequality follows from the triangular inequality. Hence, given the privacy budget $\epsilon_{\texttt{logits}} = 2$, we have $\sigma_{\texttt{logits}}=p$.
\fi

\section{More Related Works}\label{appx:morerelated}
In addition to the prior privacy-preserving machine learning works mentioned in the main paper, there are other related works in the current literature.

\textbf{Split Learning.}  Split learning \cite{SplitLearing_2018_arXiv, SplitLearning_2021, SplitFed_2022_AAAI} is another training framework targeting data protection when sharing data with other parties. 
It splits and distributes a full model between private clients and untrusted public servers. During training, the clients learn a few front layers and send intermediate representations (rather than raw inputs) to the server, relieving computation and memory pressure on clients.

However, even the intermediate representations still contain substantial sensitive information about the raw input data, which gives way to adversaries who can infer training data, especially using a model inversion attack (Section \ref{subsec:modelinversion}). While split learning can be combined with DP to ensure privacy of the intermediate representations as in \cite{SplitFed_2022_AAAI}, unfortunately, this leads to a significant accuracy drop. 
% Our work, however, does not apply DP to the intermediate representations directly to avoid this significant accuracy drop. Instead, we developed an asymmetric and rigorous way to decide how to decompose the intermediate representations and only share minimal information protected by a DP-based mechanism in public environments. Hence, our work avoids this significant accuracy drop incurred in split learning. 

\textbf{Crypto-based Private Learning.} Privacy-preserving machine learning enhanced by crypto techniques provides strong data protection \cite{DLFHE_NIPS2016, FHE_2017_Elsevier, FHE_2021_IEEE, Crypto_2018_PETS}. 
These approaches first encrypt the input data and directly train a model in the encrypted domain, preventing any untrusted parties from obtaining raw data. 
However, the encryption/decryption and bootstrapping \cite{Bootstrapping_2023} operations add tremendous complexities during training and inference, limiting their use for large-scale models.
Moreover, as non-linear activation functions are usually not supported by current encryption schemes, the crypto based solutions need to approximate these functions, which inevitably causes performance degradation. 

\textbf{TEE-based Private Learning.} Trusted execution environments (TEEs) provide a secure hardware enclave for sensitive data, which makes it a practical option for privacy-preserving machine learning \cite{TEE_2020_SP, SGX_2016, SecureTF_Mid2020, Capsule_CVPR2021, Saurav_arXiv_2020}.
TEE-based solutions encapsulate the data and the models in a secure environment and perform forward and backward passes. 
Throughout the whole process, the private information is always secured in TEEs. Therefore, such solutions achieve strong privacy protection for both data and the model. 
One critical concern of using TEEs for machine learning, however, is the relatively low computing performance compared to GPUs. Due to the low parallelism and low communication efficiency, training/inference time with TEEs differs from that with GPUs by a factor of about $100$ \cite{legrace_PETS_2022}. 
While the most advanced GPUs also come with trusted environment \cite{NvidiaTEE}, it still remains to see how this can be applied in real large-scale applications. 

\end{document}